\documentclass[11pt]{amsart}
\usepackage{geometry}                
\usepackage{graphicx}
\usepackage{amssymb}
\usepackage{epstopdf}
\usepackage{comment}
\usepackage{color}
\title{Asking the metaquestions in constraint tractability}
\author{Hubie Chen}

\address{
  University of the Basque Country (UPV/EHU),
E-20018 San Sebasti\'{a}n,
Spain
\emph{and}
IKERBASQUE, Basque Foundation for Science,
E-48011 Bilbao,
Spain
}

\author{Benoit Larose}

\address{Department of Mathematics and Statistics \\
 Concordia University \\
 1455 de Maisonneuve West\\
Montr\'eal, Qc \\
Canada, H3G 1M8} \email{benoit.larose@concordia.ca}
\urladdr{http://cicma.mathstat.concordia.ca/faculty/larose/}

\begin{document}
\newtheorem{dummy}{Dummy}[section]

\newtheorem{conj}[dummy]{Conjecture}
\newtheorem{lemma}[dummy]{Lemma}
\newtheorem{prop}[dummy]{Proposition}
\newtheorem{theorem}[dummy]{Theorem}
\newtheorem{corollary}[dummy]{Corollary}
\newtheorem{definition}[dummy]{Definition}
\newtheorem{example}[dummy]{Example}

\newcommand{\cA}{\mathcal{A}}
\newcommand{\cB}{\mathcal{B}}
\newcommand{\cC}{\mathcal{C}}
\newcommand{\cD}{\mathcal{D}}
\newcommand{\cE}{\mathcal{E}}
\newcommand{\cF}{\mathcal{F}}
\newcommand{\cG}{\mathcal{G}}
\newcommand{\cH}{\mathcal{H}}
\newcommand{\cI}{\mathcal{I}}
\newcommand{\cJ}{\mathcal{J}}
\newcommand{\cK}{\mathcal{K}}
\newcommand{\cL}{\mathcal{L}}
\newcommand{\cM}{\mathcal{M}}
\newcommand{\cN}{\mathcal{N}}
\newcommand{\cO}{\mathcal{O}}
\newcommand{\cP}{\mathcal{P}}
\newcommand{\cQ}{\mathcal{Q}}
\newcommand{\cR}{\mathcal{R}}
\newcommand{\cS}{\mathcal{S}}
\newcommand{\cT}{\mathcal{T}}
\newcommand{\cU}{\mathcal{U}}
\newcommand{\cV}{\mathcal{V}}
\newcommand{\cW}{\mathcal{W}}
\newcommand{\cX}{\mathcal{X}}
\newcommand{\cY}{\mathcal{Y}}
\newcommand{\cZ}{\mathcal{Z}}

\newcommand{\cMBW}{\mathcal{M}_{\mathsf{BW}}}

\newcommand{\bA}{\mathbf{A}}
\newcommand{\bB}{\mathbf{B}}
\newcommand{\bC}{\mathbf{C}}
\newcommand{\bD}{\mathbf{D}}
\newcommand{\bE}{\mathbf{E}}
\newcommand{\bF}{\mathbf{F}}
\newcommand{\bG}{\mathbf{G}}
\newcommand{\bH}{\mathbf{H}}
\newcommand{\bI}{\mathbf{I}}
\newcommand{\bJ}{\mathbf{J}}
\newcommand{\bK}{\mathbf{K}}
\newcommand{\bL}{\mathbf{L}}
\newcommand{\bM}{\mathbf{M}}
\newcommand{\bN}{\mathbf{N}}
\newcommand{\bO}{\mathbf{O}}
\newcommand{\bP}{\mathbf{P}}
\newcommand{\bQ}{\mathbf{Q}}
\newcommand{\bR}{\mathbf{R}}
\newcommand{\bS}{\mathbf{S}}
\newcommand{\bT}{\mathbf{T}}
\newcommand{\bU}{\mathbf{U}}
\newcommand{\bV}{\mathbf{V}}
\newcommand{\bW}{\mathbf{W}}
\newcommand{\bX}{\mathbf{X}}
\newcommand{\bY}{\mathbf{Y}}
\newcommand{\bZ}{\mathbf{Z}}
\newcommand{\BSS}{{\bf S}}


\newcommand{\bbA}{\mathbb{A}}
\newcommand{\bbB}{\mathbb{B}}
\newcommand{\bbC}{\mathbb{C}}
\newcommand{\bbD}{\mathbb{D}}
\newcommand{\bbE}{\mathbb{E}}
\newcommand{\bbF}{\mathbb{F}}
\newcommand{\bbG}{\mathbb{G}}
\newcommand{\bbH}{\mathbb{H}}
\newcommand{\bbI}{\mathbb{I}}
\newcommand{\bbJ}{\mathbb{J}}
\newcommand{\bbK}{\mathbb{K}}
\newcommand{\bbL}{\mathbb{L}}
\newcommand{\bbM}{\mathbb{M}}
\newcommand{\bbN}{\mathbb{N}}
\newcommand{\bbO}{\mathbb{O}}
\newcommand{\bbP}{\mathbb{P}}
\newcommand{\bbQ}{\mathbb{Q}}
\newcommand{\bbR}{\mathbb{R}}
\newcommand{\bbS}{\mathbb{S}}
\newcommand{\bbT}{\mathbb{T}}
\newcommand{\bbU}{\mathbb{U}}
\newcommand{\bbV}{\mathbb{V}}
\newcommand{\bbW}{\mathbb{W}}
\newcommand{\bbX}{\mathbb{X}}
\newcommand{\bbY}{\mathbb{Y}}
\newcommand{\bbZ}{\mathbb{Z}}

\newcommand{\hnote}[1]{{\color{red} [{\bf Hubie note:} #1]}}
\newcommand{\bnote}[1]{{\color{red} [{\bf Benoit note:} #1]}}

\newcommand{\csp}{\mathsf{CSP}}
\newcommand{\CSP}{\mathsf{CSP}}

\begin{abstract}
The \emph{constraint satisfaction problem (CSP)} involves deciding, 
given a set of variables and a set of constraints on the
variables, whether or not there is an assignment to the variables satisfying all of the constraints. 
One formulation of the CSP is as the problem of deciding,
given a pair $(\bbG, \bbH)$ of relational structures,
whether or not there is a homomorphism from the first structure
to the second structure.
The CSP is in general NP-hard;
a common way to restrict this problem 
is to fix the second structure $\bbH$, so that
each structure $\bbH$ gives rise to a problem $\csp(\bbH)$.
The problem family $\csp(\bbH)$ has been studied
using an algebraic approach, which links the algorithmic and
complexity properties of each problem $\csp(\bbH)$
to a set of operations, the so-called \emph{polymorphisms}
of $\bbH$. 
Certain types of polymorphisms are known to imply the
polynomial-time tractability of $\csp(\bbH)$, and
others are conjectured to do so.
This article systematically studies---for various classes
of polymorphisms---the computational complexity
of deciding whether or not a given structure $\bbH$
admits a polymorphism from the class. Among other results, we prove
the NP-completeness of deciding a condition conjectured
to characterize the tractable problems $\csp(\bbH)$,
as well as the NP-completeness of deciding
if $\csp(\bbH)$ has bounded width.
\end{abstract}

\maketitle

\section{Introduction}

The \emph{constraint satisfaction problem (CSP)} involves deciding, 
given a set of variables and a set of constraints on the
variables, whether or not there is an assignment to the variables satisfying all of the constraints. Cases of the constraint
satisfaction problem appear in many fields of study, including artificial intelligence, spatial and temporal reasoning,
logic, combinatorics, and algebra. Indeed, the constraint satisfaction problem is flexible in that it admits a number of
equivalent formulations. 
In this paper, we work with the well-known formulation 
as the relational homomorphism problem, namely: given
two similar relational structures $\bbG$ and $\bbH$, does there exist a homomorphism from $\bbG$ to $\bbH$? In this formulation, one
can view each relation of $\bbG$
as containing variable tuples that are constrained together, and the corresponding relation
of $\bbH$ as containing the permissible values for the variable tuples.
In this article, we assume that all structures under discussion are finite,
that is, have finite universe.

The constraint satisfaction problem is in general NP-hard; this general intractability has motivated the study of
restricted versions of the CSP that have various desirable complexity and algorithmic properties. A natural and well-studied
way to restrict the CSP is to fix the second structure $\bbH$
(often referred to as the \emph{right-hand side structure}),
which amounts to restricting the relations that can be used
to specify permissible value tuples.
Each structure $\bbH$ then gives rise to a problem
$\csp(\bbH)$: given a structure $\bbG$, decide if
it has a homomorphism to $\bbH$; 
and, 
the resulting family of problems is a rich one that includes
Boolean satisfiability problems, graph homomorphism problems, and satisfiability problems on algebraic equations.
While each problem $\csp(\bbH)$ is in NP, 
for certain structures $\bbH$ it can be shown that the problem
$\csp(\bbH)$ is polynomial-time decidable.
Indeed, in a now-classic result from 1978, 
Schaefer~\cite{Schaefer78-satisfiability} 
presented a classification theorem,
showing that for
each structure $\bbH$ having a two-element universe,
the problem $\csp(\bbH)$ is either polynomial-time decidable,
or is NP-hard.
Schaefer left open and suggested the research program of
classifying structures having finite universe of size strictly
greater than two.

Over the past two decades, an algebraic approach to studying
complexity aspects of
the problem family
$\csp(\bbH)$
has emerged.
A \emph{polymorphism} of a structure $\bbH$ with universe $H$
is defined as a finitary operation $f: H^k \to H$ that is a homomorphism
from $\bbH^k$ to $\bbH$; note that a polymorphism of arity $k=1$
is precisely an endomorphism.
A cornerstone of the algebraic approach is a theorem stating
that when two structures $\bbH$, $\bbH'$ have the same polymorphisms,
the problems $\csp(\bbH)$ and $\csp(\bbH')$ are polynomial-time
interreducible~\cite{BulatovJeavonsKrokhin05-finitealgebras}.\footnote{In fact,
under the stated assumption, the problems $\csp(\bbH)$ and $\csp(\bbH')$
are logarithmic-space interreducible~\cite{LaroseTesson09-hardness}.
Let us mention here that, under the assumption,
one also has interreducibility for some other computational problems
of interest, such as
the \emph{quantified CSP}~\cite{BBCJK09-qcsp,Chen11-qcspandpgp,Chen12-meditations}
and various comparison problems involving primitive positive
formulas~\cite{BovaChenValeriote13-generic}.
}  
Intuitively, this theorem
can be read as saying that the polymorphisms of a structure
contain all of the information one needs to know to understand
the complexity of $\csp(\bbH)$, 
at least up to polynomial-time computation.
At the present, it is well-known that certain types of polymorphisms
are desirable in that they
guarantee polynomial-time tractability of $\csp(\bbH)$.
As an example, it is now a classic theorem in the area
that, for any structure $\bbH$ having a \emph{semilattice polymorphism},
the problem $\csp(\bbH)$ is polynomial-time decidable;
a semilattice polymorphism is, by definition, an arity $2$ polymorphism
that is associative, commutative and idempotent.
Here, it should be further pointed out that
a conjecture known as the 
\emph{algebraic dichotomy conjecture}~\cite{BulatovJeavonsKrokhin05-finitealgebras}
predicts the polynomial-time tractability of 
each problem $\csp(\bbH)$ not satisfying
a known sufficient condition for NP-completeness,
and that this conjecture can be formulated as predicting
the tractability of each problem $\csp(\bbH)$ 
where $\bbH$ admits
a certain type of polymorphism
(see 
Conjecture~\ref{conj:algebraic-dichotomy}
and the surrounding discussion).


In this article, 
we systematically study---for various 
classes of polymorphisms---the computational problem
of deciding whether or not a given structure
$\bbH$ admits a polymorphism from the class.
This form of decision problem is often popularly referred to 
as a \emph{metaquestion}.
All of the polymorphisms that we study are either known 
to guarantee tractability of $\csp(\bbH)$, 
or predicted to do so by the algebraic dichotomy conjecture
(see the discussion in Section~\ref{sect:defs}).

Let us overview our principal technical results.
\begin{itemize}
\item We formalize and demonstrate a connection between the polynomial-time
tractability of a particular type of metaquestion
and the existence of a so-called uniform
polynomial-time algorithm for the condition that the metaquestion
asks about (Section~\ref{sect:uniformity-and-metaquestions}).

\item On the positive side, we prove that the metaquestion
for conservative binary commutative polymorphisms is
solvable in NL, non-deterministic logspace
(Section~\ref{subsect:cc-polymorphisms}).

\item We prove a generic NP-hardness result that applies
to the metaquestions corresponding to a range of Maltsev conditions
(Section~\ref{subsect:generic-hardness}).
One consequence of this result is that deciding if a given
structure gives rise to a CSP with bounded width is NP-complete
(Corollary~\ref{cor:deciding-bounded-width});
this answers a question of L. Barto~\cite{Barto14-collapse}.
Another consequence of this result is the
NP-completeness of
deciding if a given structure satisfies an algebraic condition
which has been conjectured
 to characterize the structures having a tractable CSP
(see Corollary~\ref{cor:deciding-quasi-siggers}).

\item We provide a simple proof
 that the metaquestion for semilattice polymorphisms
is NP-complete
(Section~\ref{subsect:semilattice}).

\item We give a general hardness result showing that, 
for a number of types of conservative polymorphisms, 
the metaquestion is NL-hard
(Section~\ref{subsect:nl-hardness}).
In particular, this result applies to the metaquestion
for conservative binary commutative polymorphisms,
and hence provides a hardness result tightly complementing
the positive result for such polymorphisms.

\end{itemize}
We summarize some consequences both of our results and known results in Table~\ref{table:summary}.

We view the complexity study of metaquestions as a naturally motivated
research topic.  
In general,
an instance $(\bbG, \bbH)$ of the CSP encountered in the wild 
or on the street
does not, of course, 
come with any guarantee about the
properties of the right-hand side structure $\bbH$;
in order to know if any of the polymorphism-based
tractability results can be exploited to solve the instance,
one must first detect if $\bbH$ has a relevant polymorphism.
From this perspective, the present study can thus be viewed
as an effort to bridge practice and the algebraic theory
of tractability.




\begin{center}
\begin{table}
\begin{tabular}[t]{|c||c|c|c|c|} 
\hline
 Polymorphism&free &idempotent &conservative&conservative  \\
 & & &&at most binary structure  \\
  \hline &&&& \\
  2-TS & NP-c &(*)                        &  NL-c      & NL-c \\
 \hline &&&& \\
$k$-TS ($k \geq 3$) &NP-c&P               &  P  &P/NL-hard \\
 \hline &&&& \\
$k$-symmetric ($k \geq 3$, even) &NP-c&(*)&  P  &  P/NL-hard \\
 \hline &&&& \\
 $k$-symmetric ($k \geq 3$, odd) &NP-c&(*)& not known &P \\
  \hline &&&& \\
$k$-cyclic ($k \geq 3$, even) &NP-c&(*)   & P   &  P/NL-hard \\
  \hline &&&& \\
  $k$-cyclic ($k \geq 3$, odd) &NP-c&(*)& not known &P \\
  \hline &&&& \\
 Set polymorphism & EXPTIME/ & EXPTIME & EXPTIME & EXPTIME/ \\
  &NP-hard  & & & NL-hard \\
  \hline &&&& \\
   Maltsev &(*)&(*)&P \cite{Frenchguypaper} &P \\
  \hline &&&& \\
    Siggers &(*)&(*)& not known &P \\
  \hline &&&& \\
 semilattice &NP-c&NP-c&NP-c&NP-c \\
  
\hline
\end{tabular}

\caption{{\bf Summary of some results.}
An entry of (*) indicates that the existence of the type of polymorphism
in question can be formulated as an
idempotent strong linear Maltsev condition,
implying the applicability of
Corollary~\ref{cor:metaquestion-uniformity},
which connects the metaquestion for the Maltsev condition
to the existence of a uniform polynomial-time algorithm
(see Section~\ref{sect:uniformity-and-metaquestions} for details).
Note that positive results propagate to the right, and
hardness results propagate to the left; for readability,
we omit explicitly placing NL-hardness claims in some of the entries.
The NP-completeness results on semilattices
come from Theorem~\ref{thm:semilattice-hard};
the other NP-completeness results and the NP-hardness result
come from Corollary~\ref{cor:deciding-ts-symm-cyc}.
All NL-hardness results come from Theorem~\ref{thm-nl-hard},
and NL containment for conservative $2$-TS 
 comes from 
Theorem~\ref{thm:cc-nl}.
The EXPTIME containment result for set polymorphisms
comes from Proposition~\ref{prop:set-polymorphism-exptime}.
The P containment result for $k$-TS comes from 
Corollary~\ref{cor:idempotent-ts-p},
and the P containment results for even $k$-symmetric
and even $k$-cyclic come from Corollary~\ref{cor:even-conservative}.
The P containment result for conservative Maltsev
comes from
\cite{Frenchguypaper}.
Finally, in the case of conservative at most binary structures,
the not yet covered P containment results follow from
Theorem~\ref{thm:at-most-binary-ptime} 
in conjunction with
Proposition~\ref{prop:polymorphism-to-siggers}.  
}
\label{table:summary}
\end{table}

\end{center}
\vspace{.5cm}

\section{Definitions, Notation and Terminology}
\label{sect:defs}

A {\em relational structure} 
is a tuple $\bbH = \langle H; \theta_1,\dots,\theta_s \rangle$ where $H$ is a non-empty finite set and each $\theta_i$ is a  relation of  arity $r_i$ on $H$; the sequence $r_1,\dots,r_s$ is the {\em type} of $\bbH$. 
A relational structure is \emph{at most binary} if
the arity of each relation is less than or equal to $2$.
In this article, most of the computational problems considered
take as input a relational structure; 
as is quite standard in the literature,
\emph{we always assume that each relation of a relational structure
is specified by an explicit listing of its tuples.}
Two structures with the same type are said to be {\em similar}. If $\bbG$, $\bbH$, $\bbK$, ... are relational structures, we denote their respective universes by $G$, $H$, $K$, ... The {\em product} of similar structures is the usual one, viz. if    $\bbG = \langle G; \theta_1,\dots,\theta_s \rangle$ and $\bbH = \langle H; \rho_1,\dots,\rho_s \rangle$ then $\bbG \times \bbH = \langle G \times H, \sigma_1,\dots,\sigma_s \rangle$ where $\sigma_i = \{((g_1,h_1),\dots,(g_r,h_r)): (g_1,\dots,g_r) \in \theta_i, (h_1,\dots,h_r) \in \rho_i\}$. We denote the product of the structure $\bbH$ with itself $k$ times by $\bbH^k$. Given a map $f:G \rightarrow H$ and a $k$-tuple $u = (u_1,\dots,u_k) \in G^k$, let $f(u) = (f(u_1),\dots,f(u_k))$; if $\theta$ is a $k$-ary relation on $\bbG$ then $f(\theta) = \{(f(u): u \in \theta\}$.
A map $f:G\rightarrow H$ is a {\em homomorphism from $\bbG$ to $\bbH$ } if $f(\theta_i) \subseteq \rho_i$ for all $i=1,\dots,s$.
For an integer $k \geq 1$, a {\em  $k$-ary operation on $H$} is a map from $H^k$ to $H$. 

\begin{definition} Let $\bbH$ be a relational structure. 
A $k$-ary operation $f$ on $H$ is a {\em polymorphism} of $\bbH$ if $f$ is a homomorphism from $\bbH^k$ to $\bbH$;
in this case, we also say that $f$ \emph{preserves} $\bbH$.
\end{definition}

We are concerned with polymorphisms obeying various interesting identities. In order to avoid undue algebraic technicalities, 
we present certain concepts in a slightly unorthodox way 
(for the standard equivalents, see for instance \cite{KozikKrokhinValerioteWillard15-characterizations}.)

An expression $\cE$ of the form 
$$f(x_1,\dots,x_k) \approx g(y_1,\dots, y_n)$$
is a {\em linear identity};
it is satisfied by two interpretations for $f$ and $g$ on
a set $H$ if, for any assignment to the variables,
it holds that both sides of the identity evaluate to the same value.
Without fear of confusion, we shall usually make no distinction between the operation symbols used in an identity and the actual operations satisfying it, for example, we will simply write that $f$ and $g$ satisfy $\cE$ and so on.
Note that we allow linear identities of the form 
$$f(x_1,\dots,x_k) \approx y_i$$
that is, containing only one operation symbol;
such an identity can be formally viewed as an
expression of the above form
where one of the operations is a projection. Following 
\cite{BartoOprsalPinsker15-wonderland}, if a linear identity is not of this form, i.e. has explicit operation symbols on both sides, we say it has {\em height 1}.

A {\em strong linear Maltsev condition} is a finite set of linear identities $\{\cE_1,\dots,\cE_r\}$. A sequence of operations $f_1,\dots,f_m$ 
{\em satisfies} the strong Maltsev condition if it satisfies each identity $\cE_i$.\footnote{
 The standard definition of Maltsev conditions concerns varieties of algebras. The modifier ``strong'' refers  to the fact that set of identities is  finite, as opposed to a condition as in Lemma \ref{lemma-BKS} below.

}

\begin{definition} Let $\bbH$ be a relational structure. We say that $\bbH$ {\em satisfies}  a strong linear Maltsev condition if there exist polymorphisms of $\bbH$ that satisfy it. \end{definition}


We now present some strong Maltsev conditions we shall investigate.

A $k$-ary operation $f$ (with $k \geq 1$)
is {\em idempotent} if it satisfies $$f(x,x,\dots,x) \approx x.$$ 
The operation $f$ is {\em cyclic} if it obeys 
$$f(x_1,\dots,x_k) \approx f(x_k,x_1,\dots,x_{k-1}). $$ 
It is {\em symmetric} if,
for every permutation $\sigma$ of the set $\{1,\dots,k\}$,
it obeys the identities
$$f(x_1,\dots,x_k) \approx f(x_{\sigma(1)},\dots,x_{\sigma(k)}).$$
It is {\em totally symmetric (TS)} if, whenever
$\{x_1,\dots,x_k\} = \{y_1,\dots,y_k\}$,
it satisfies the identity
$$f(x_1,\dots,x_k) \approx f(y_1,\dots,y_k).$$

 Notice that 
$$f \text{ TS } \Rightarrow f \text{ symmetric } \Rightarrow f \text{ cyclic }$$ and that, for a binary operation, the properties of
being commutative, TS, symmetric, and cyclic all coincide.

For $k \geq 3$, the operation $f$ is a {\em near-unanimity (NU) operation} if it obeys the identities
$$f(x,\dots,x,y,x,\dots,x) \approx x$$ for any position of the lone $y$. A 3-ary NU operation is called a {\em majority} operation. 

A 3-ary operation $f$ is {\em Maltsev} if it obeys the identities
$$f(y,y,x) \approx f(x,y,y) \approx x.$$

A 4-ary operation $f$ is {\em Siggers} if it is idempotent and satisfies the identity
$$f(a,r,e,a) \approx f(r,a,r,e).$$

We shall also require the following conditions on operations, 
which are {\em not} presented by linear identities. The $k$-ary operation $f$ is {\em conservative} if it satisfies $$f(x_1,\dots,x_k) \in \{x_1,\dots,x_k\}$$ for all $x_i$. A {\em semilattice} operation is an associative, idempotent, commutative binary operation. 

We now gather some well-known implications involving the special polymorphisms defined here; as some of these results are folklore, we give general references only \cite{BulatovValeriote08-recent-results-csp,Siggers2010,Kearnes2014,KozikKrokhinValerioteWillard15-characterizations}. 

\begin{prop} 
\label{prop:polymorphism-to-siggers}
If a structure admits an idempotent polymorphism $f$ 
which is cyclic, TS, symmetric, NU, or Maltsev 
then it admits a Siggers polymorphism; 
moreover, in each case, if $f$ is conservative, so is the Siggers polymorphism. 
\end{prop}

\section{Known/Preliminary Results} \label{section-known}

Let $\bbH$ be a relational structure. We denote by $\CSP(\bbH)$ the set of finite structures that admit a homomorphism to $\bbH$. The problem
 $\CSP(\bbH)$ is clearly in NP.
The dichotomy conjecture of Feder and Vardi, that states that every $\CSP(\bbH)$ is either tractable or NP-complete, has been the source of intense scrutiny over the past two decades, see for instance \cite{BartoKozik09-boundedwidth, BIMMVW10-fewsubpowers} and the surveys \cite{DBLP:conf/dagstuhl/BulatovKL08,BulatovValeriote08-recent-results-csp}.
A very deep theory has been developed, relating the nature of the identities satisfied by the polymorphisms of the structure $\bbH$ and the complexity of the associated constraint satisfaction problem. Simplifiying to the extreme, the theory states that, the nicer the identities, the easier the problem is.  

We say the structure $\bbR$ is {\em a retract} of the structure $\bbH$ if there exist homomorphisms $r:\bbH \rightarrow \bbR$ and $e:\bbR \rightarrow \bbH$ such that $r \circ e$ is the identity on $R$. 
A structure $\bbH$ is a {\em core} if the only homomorphisms from $\bbH$ to itself are automorphisms, or equivalently, if the structure has no proper retract. 
The retracts of minimal size of a finite relational structure $\bbH$ are cores, and are all isomorphic to each other; we refer to these retracts
as the \emph{cores of $\bbH$}, and due to their being mutually isomorphic, by a slight abuse we speak of
{\em the core} of a structure.
 Obviously if $\bbH'$ is the core of $\bbH$ then $\CSP(\bbH) = \CSP(\bbH')$. 
It is known that
for a core $\bbH$, the problem $\CSP(\bbH)$ is interreducible with
the problem $\CSP(\widehat{\bbH})$ where the structure $\widehat{\bbH}$ 
is obtained by expanding the structure $\bbH$ with all one-element
unary relations (sometimes called \emph{constants} in this context).
Here, interreducibility can actually be proved
with respect to first-order reductions~\cite{LaroseTesson09-hardness}. 
As such a structure $\widehat{\bbH}$ has only idempotent polymorphisms
(indeed, it is straightforwardly verified that
the polymorphisms of $\widehat{\bbH}$ are precisely
the idempotent polymorphisms of $\bbH$),
for many complexity issues on the problem family $\CSP(\bbG)$,
one can restrict attention to idempotent algebras,
which are known to have good behavior.
Note also that, up to logspace interreducibility,  we can assume the equality relation is also a relation of 
a structure $\bbH$ where $\csp(\bbH)$ is under study \cite{BulatovJeavonsKrokhin05-finitealgebras}. Finally,  the following is immediate: a structure $\bbH$ admits a conservative operation $f$ satisfying some identities if and only if the structure $\bbH'$ obtained from $\bbH$ by adding all non-empty subsets as basic relations admits an operation satisfying those same identities.

The following is one of many possible formulations of a refinement of the dichotomy conjecture, due to Bulatov, Jeavons and 
Krokhin~\cite{BulatovJeavonsKrokhin05-finitealgebras} (see also \cite{KozikKrokhinValerioteWillard15-characterizations}):


\begin{conj}  
\label{conj:algebraic-dichotomy}
If the core of a relational structure $\bbH$ admits a Siggers polymorphism,
then $\CSP(\bbH)$ is tractable. \end{conj}

A form of converse to this statement is known to hold, namely,
it holds that
a  structure whose core has no Siggers polymorphism 
has an NP-complete CSP \cite{BulatovJeavonsKrokhin05-finitealgebras}. 


Let us remark here
that the conservative case was completely settled by Bulatov:

\begin{theorem} \cite{Bulatov11-conservative-csp}
 If  a relational structure admits a conservative Siggers polymorphism then its CSP is tractable. 
\end{theorem}

In the rest of this section, we will focus on CSP's
satisfying a condition called \emph{bounded width}.
Most known tractable CSP's can be grouped roughly into two distinct families: 
bounded width and few subpowers. 
Few subpowers problems~\cite{Chen05-expressive,IMMVW10-tractabilityfewsubpowers}
generalize
linear equations and are solvable by an algorithm with many properties in common with Gaussian elimination. 
In particular, structures admitting a Maltsev or near-unanimity polymorphism have this property. 
However, in general,
 the algorithms involved require explicit knowledge of the polymorphisms
that witness the condition of few subpowers. 
In contrast, recent results on CSPs of bounded width show that they are actually solvable by an algorithm that is uniform in the sense that 
it needs no such explicit knowledge of the polymorphisms.
This form of uniformity has important consequences for the metaproblem. We now discuss this in more detail.

In order to present the required algorithm, 
it will convenient to view CSPs in a slightly different way; the fact  that both approaches are equivalent is well-known and easy to verify.
We essentially follow \cite{Barto14-collapse}.
Let us introduce some notation.  
If $f$ is a function with domain $D$ and $W \subseteq D$, let $f|_W$ denote the restriction of $f$ to $W$. Similarly, if $C \subseteq H^D$ is a family of functions with domain $D$, let $C|_W = \{f|_W: f \in C\}$. \\

An {\em instance } of the CSP is a triple 
$\cI = (V,H,\cC)$ where 
\begin{itemize}
\item V is a non-empty, finite set of {\em variables},
\item $H$ is a nonempty finite set of {\em values},
\item $\bC$ is a finite nonempty set of {\em constraints}, where each constraint is a subset $C$ of $H^W$;
 $W$ is a subset of $V$ called the {\em scope} of the constraint, and $|W|$ is called the {\em arity} of the constraint.
\end{itemize}

A {\em solution} of the instance is a map $f:V \rightarrow H$ such that, for every constraint $C$ with scope $W$, we have $f|_W \in C$. 

If $W = \{x_1,\dots,x_k\}$ we can associate naturally a $k$-ary relation $\theta$ to each subset $C$ of $H^W$  by setting $\theta = \{f(x_1),\dots,f(x_k)): f \in C\}$ (this depends of course on the ordering of $W$ chosen.) In this way, one can restrict the nature of the constraints involved in instances by stipulating that their associated relations belong to some fixed set; it follows that we can view the problem $\CSP(\bbH)$ as a set of instances of the form just described.

Let $1 \leq k \leq l$. Consider the following polynomial-time algorithm that transforms an instance into a so-called {\em $(k,l)$-minimal instance}:

\noindent{\bf The $(k,l)$-minimality algorithm.}
\begin{itemize}
\item For each $l$-element set $W \subseteq V$, add a ``dummy'' constraint $H^W$ (this is to ensure every $l$-element set of variables is contained in the scope of some constraint);

\item Repeat the following process until it stabilises: for every subset $W \subseteq V$ of size at most $k$, and every pair of constraints $C_1$ and $C_2$ whose scope contains $W$, remove from $C_1$ and $C_2$ any function such that $f|_W \not\in C_1|_W \cap C_2|_W$.
\end{itemize}
It is easy to see that the instance obtained is equivalent to the original, in the sense that they have the same solutions. In particular, if the output instance has an empty constraint then the original instance had no solution. On the other hand, if one does not obtain an empty constraint, there is no guarantee the original instance has a solution. 

\begin{definition} The problem $\CSP(\bbH)$ has 
{\em relational width $(k,l)$} 
if the $(k,l)$-minimality algorithm correctly decides it,
that is, if the $(k,l)$-minimality algorithm
detects an empty constraint whenever the input is a 
\emph{no} instance.
We say the problem $\CSP(\bbH)$ 
has {\em bounded width} 
if it has relational width $(k,l)$ for some $1 \leq k \leq l$. 
\end{definition}

The cores $\bbH$ whose CSP have bounded width were characterized by Barto and Kozik~\cite{BartoKozik14-localconsistency};
the following description is known.

\begin{definition}
\label{def:bw-operations}
Let us say that a pair of operations $v,w$ are
\emph{BW operations} if $v$ is $3$-ary, $w$ is $4$-ary,
and they satisfy the following identities:
\begin{itemize}

\item $v(y,x,x) \approx w(y,x,x,x)$
\item $v(y,x,x) \approx v(x,y,x) \approx v(x,x,y)$
\item $w(y,x,x,x) \approx w(x,y,x,x) \approx w(x,x,y,x) \approx w(x,x,x,y)$

\end{itemize}
\end{definition}

\begin{theorem} \label{thm:bw-cores}
\cite{BartoKozik14-localconsistency,KozikKrokhinValerioteWillard15-characterizations}
Let $\bbH$ be a core.  The problem $\csp(\bbH)$ has bounded width
if and only if $\bbH$ 
has idempotent polymorphisms $v$ and $w$ that are
BW operations.
\end{theorem}

It turns out that problems of bounded width are precisely those solvable by a Datalog program; and the problems solvable by a monadic Datalog program are precisely those of relational width $(1,l)$ for some $l$, i.e. problems of {\em width 1}. 
Let us here observe the following fact.

\begin{prop}
\label{prop:bw-preserved-by-hom-equivalence}
Let $\bbH$ and $\bbH'$ be homomorphically equivalent structures.
The problem $\csp(\bbH)$ has bounded width if and only if
the problem $\csp(\bbH')$ does.  
In particular, for any structure $\bbH$,
let $\bbC$ denote its core;
the problem $\csp(\bbH)$ has bounded width if and only if
the problem $\csp(\bbC)$ does.
\end{prop}

One way to prove this fact is to use the just-mentioned characterization
of bounded width via Datalog programs; a Datalog program only
processes the input structure $\bbG$ (and not the right-hand side structure), so the fact follows from this characterization
and the fact that $\csp(\bbH) = \csp(\bbH')$
(when $\bbH, \bbH'$ are homomorphically equivalent).

The following result due to Barto also implies the Datalog hierarchy collapses, in that every problem of bounded width can be solved using the $(2,3)$-minimality algorithm.

\begin{theorem}[\cite{Barto14-collapse}] \label{thm-barto-collapse} For every relational structure $\bbH$, exactly one of the following holds:
\begin{enumerate}
\item  $\CSP(\bbH)$ has width 1;
\item  $\CSP(\bbH)$ has relational width $(2,3)$ but not width 1;
\item  $\CSP(\bbH)$ does not have bounded width.
\end{enumerate} 
\end{theorem}

 The next result is a slight generalisation of Corollary 8.4 of 
 \cite{Barto14-collapse}; we include its proof here as it is quite simple, clever and instructive.
 Recall from Section \ref{section-known} that the structure $\widehat{\bbH}$ 
is obtained by expanding the structure $\bbH$ with all one-element
unary relations.

\begin{lemma} \label{bwalgo} (Barto, Kozik, Maroti, unpublished) Let  $\{\cE_1,\dots,\cE_r\}$ be a strong linear  Maltsev condition and let $\cC$ be a class of structures such that, if $\bbH \in \cC$ is a structure satisfying  $\{\cE_1,\dots,\cE_r\}$ then $\CSP(\widehat{\bbH})$ has bounded width. 
There exists a polynomial-time algorithm that,
given as input a structure in $\cC$, decides if the structure
satisfies $\{\cE_1,\dots,\cE_r\}$.
\end{lemma}

\begin{proof} Given a structure $\bbH \in \cC$, we set up an instance of $\CSP(\widehat{\bbH})$ to encode the existence of the required polymorphisms as follows: the universe of our instance  is the disjoint union of $\bbH^{r_1},\dots,\bbH^{r_s}$ where the $r_i$ are the arities of the different operation symbols in the Maltsev condition. We add equality constraints for the required identities, i.e. if $f(x_1,\dots,x_k) \approx g(y_1,\dots,y_n)$ is an identity of our condition, we identify every pair of tuples satisfying it in the copies corresponding to $f$ and $g$. Similarly, we add the necessary unary constraints corresponding to identities of the form $f(x_1,\dots,x_k) \approx y_i$. Now we run the (2,3)-consistency algorithm on this instance. If it answers no, then $\bbH$ does not satisfy the condition (since by hypothesis if it did the consistency algorithm could not give a false positive). Otherwise, select an element of the instance, and a value of $H$, fixing the value of the solution on the element to this value. We run the (2,3)-consistency algorithm on this new instance. Looping on all possible values of $H$, we reject if there is no satisfying value. Otherwise, we keep this value, and repeat the procedure with all other elements of the instance. 
If this process terminates without rejecting, we have a fully-defined sequence of operations satisfying the condition. 
 \end{proof}
 
We observe the following corollaries of the preceding lemma, combined with the following result.

Let $t$ be a $k$-ary operation on the set $H$ and let $A$ be a $k \times k$ matrix with entries in $H$. We write $t[A]$ to denote the $k \times 1$ matrix whose entry on the $i$-th row is the value of $f$ applied to row $i$ of $A$. 
 \begin{lemma}[\cite{BartoKozikStanovsky15-maltsev-lack-solvability}, Theorem 1.3] \label{lemma-BKS}
 Let $t$ be a $k$-ary idempotent polymorphism of $\bbH$ satisfying: $t[A] = t[B]$ where $A$ and $B$ are $k \times k$ matrices with entries in $\{x,y\}$ such that $a_{ii} = x$ and $b_{ii} = y$ for all $i$ and $a_{ij}=b_{ij}$ for all $i \neq j$. Then $\CSP(\widehat{\bbH})$ has bounded width. \end{lemma}

 \begin{corollary} 
 There exists a polynomial-time algorithm that, given a core structure $\bbH$,
 decides if $\csp(\bbH)$ has bounded width.
 \end{corollary}

 \begin{proof}
Let $\cM$ be the strong linear Maltsev condition that 
asserts that $v$ and $w$ are idempotent BW operations.
Note that for any core $\bbH$, the structure $\widehat{\bbH}$ is also a core, and that 
$\bbH$ satisfies $\cM$ if and only if 
$\widehat{\bbH}$ satisifes $\cM$.
Let $\cC$ be the class of all cores.
Then, the hypothesis of Lemma~\ref{bwalgo} is satisfied:
for any core $\bbH$ satisfying $\cM$,
it holds that $\widehat{\bbH}$ is also a core that satisfies $\cM$,
and thus that $\csp(\widehat{\bbH})$ has bounded width
by Theorem~\ref{thm:bw-cores}.
The result of Lemma~\ref{bwalgo} thus gives the desired algorithm,
since a core has bounded width if and only if it satisfies
$\cM$, by Theorem~\ref{thm:bw-cores}.
 \end{proof}

 \begin{corollary}
 \label{cor:knu-polytime}
For each $k \geq 3$,
there exists a polynomial-time algorithm that decides
if a given structure has a $k$-near unanimity polymorphism.
 \end{corollary}

\begin{proof}
The presence of such a polymorphism is formulable as a 
strong linear Maltsev condition.
If a structure $\bbH$ has a $k$-near unanimity polymorphism,
so does $\widehat{\bbH}$ 
(as such a polymorphism is idempotent by definition); in this case,
it is known that $\csp(\widehat{\bbH})$ has bounded width. This was proved first in ~\cite{JeavonsCohenCooper98-consistency}, and also follows from Lemma~\ref{bwalgo}, 
using the following matrices: $A$ is the matrix with all $x$'s, and $B$ is obtained from $A$ by placing $y$'s on the diagonal.   The result thus follows from Lemma~\ref{bwalgo}.
\end{proof}

 \begin{corollary} 
\label{cor:idempotent-ts-p}
 Let $k \geq 3$. 
 Deciding if a structure admits a  $k$-ary idempotent TS polymorphism is in P. \end{corollary}
 
 \begin{proof}  Let $f$ be a $k$-ary idempotent TS operation, $k \geq 3$. Then it satisfies the following identities:
 \begin{align*}
 f(x,y,x,\dots,x) &\approx f(y,y,x,\dots,x)\\
  f(x,x,y,\dots,x) &\approx f(x,y,y,x,\dots,x)\\
  \cdots \\
  f(y,x,\dots,x,x) &\approx f(y,x,\dots,x,y)
 \end{align*}
 and hence by the last lemma the conditions of Lemma~\ref{bwalgo} are satisfied. 
 \end{proof}
 
The following lemma was first proved by Bulatov and Jeavons. We include a streamlined proof, as we believe that it may be of independent interest.
 
\begin{lemma} \label{types} 
\cite{BulatovJeavons00-commutative-conservative}
Let $\bbH$ be a relational structure that admits a binary conservative commutative polymorphism. Then $\CSP(\bbH)$ has bounded width. \end{lemma}

\begin{proof} The proof follows the very same strategy as case (2) of Theorem 4.1 of \cite{LaroseTesson09-hardness}; 
we refer the reader to this proof for full details. Suppose $\bbH$ has a binary conservative commutative polymorphism $t$ but $\CSP(\bbH)$ does not have bounded width. 
Combining the results of \cite{BartoKozik09-boundedwidth},  \cite{Szendrei92-survey}, and  \cite{Valeriote09-subalgebra-intersection}, there exists a  subset $C$ of $H$, $|C| \geq 2$, such that that every polymorphism of $\bbH$ restricted to $C$ preserves the relation $\rho = \{(a,b,c):a+b=c\}$ for some Abelian group structure on $C$.  Pick an element $c \neq 0$ in $C$. Then $(0,c,c),(c,0,c) \in \rho$ implies that $t(0,c) + t(c,0) = t(c,c) = c$. Since $c\neq 0$ and $t$ is commutative and conservative it follows that $t(0,c)=t(c,0)=c$. But then $c + c = c$ which implies $c=0$, a contradiction. \end{proof}

\begin{corollary} 
\label{cor:even-conservative}
Let $k$ be a positive even integer. 
There is a polynomial-time algorithm to decide whether a relational structure $\bbH$ admits a conservative, cyclic polymorphism of arity $k$;
the same result holds for conservative, symmetric polymorphisms.
 \end{corollary}

\begin{proof} We prove the result for  cyclic polymorphisms;
 the symmetric case is identical. Determining if a structure $\bbH$ admits a conservative cyclic polymorphism of arity $k$ is clearly equivalent to determining if the structure $\bbH'$ obtained from $\bbH$ by adding all non-empty subsets of the universe as unary relations admits a $k$-ary cyclic polymorphism. 
Furthermore, notice that if a structure admits a 
conservative cyclic polymorphism $f$ of even arity, then it admits one of arity 2 by identifying variables: $g(x,y) = f(x,\dots,x,y,\dots,y)$ in the obvious way. Thus by Lemma~\ref{types}, 
we can apply  Lemma \ref{bwalgo} with $\cC$  the class of conservative structures.  \end{proof}

 It is interesting to note the following consequence of Lemma~\ref{bwalgo}: on the class of digraphs, deciding the existence of a Maltsev polymorphism is tractable.
 Indeed, by a result of Kazda \cite{Kazda11-maltsev-digraphs-have-majority},
 if a digraph admits a Maltsev polymorphism it also has a majority polymorphism; since the existence of a near-unanimity polymorphism implies the CSP has bounded width 
 (this is as an easy exercise using  Lemma~\ref{lemma-BKS}), the result follows from Lemma~\ref{bwalgo}. 

 As another example of a class of structures where Lemma~\ref{bwalgo} can be invoked, consider 
the class of conservative, at most binary structures, i.e.  structures whose basic relations are at most binary and include all non-empty subsets of the universe as unary relations. 
Kazda has proved~\cite{Kazda11-csp-binary-conservative} 
that if such a structure admits a Siggers polymorphism 
(i.e. if the CSP is tractable), then in fact the CSP has bounded width; 
the following is a consequence of this and Lemma~\ref{bwalgo}.

\begin{theorem} 
\label{thm:at-most-binary-ptime} Let $\cM$ be a strong linear Maltsev condition such that, if a structure $\bbH$ admits a conservative polymorphism satisfying $\cM$ then $\bbH$ admits a Siggers polymorphism. 
Then it is polynomial-time decidable if  an at most binary structure admits a
conservative polymorphism satisfying $\cM$.

\end{theorem}

\section{Uniformity and metaquestions}
\label{sect:uniformity-and-metaquestions}

We saw in the previous section that $(2,3)$-consistency
is a generic polynomial-time algorithm that 
uniformly solves all problems $\csp(\bbH)$
of bounded width. 
In this section, we formalize and study notions of uniform
polynomial-time algorithms, for a Maltsev condition,
and observe a direct relationship 
between the existence of a uniform polynomial-time algorithm,
for a Maltsev condition,
and the corresponding metaquestion for the condition
(see Theorem~\ref{thm:metaquestion-and-uniformity}
for a precise statement).
This relationship, as will be made evident,
generalizes Lemma~\ref{bwalgo} in a certain sense.
We also mention that it seems to have been a matter of folklore
that such a relationship held; in particular,
the consequence noted in 
Example~\ref{ex:maltsev-equivalence} below
was previously communicated to us by M. Valeriote.



\begin{definition} Let $\cC$ be a class of finite structures. 
A {\em uniform polynomial-time algorithm} for $\cC$
is
a polynomial-time algorithm 
that, for each CSP instance $(\bbG, \bbH)$ with $\bbH \in \cC$,
correctly decides the instance.
\end{definition}

\begin{example}
Barto's collapse result (Theorem \ref{thm-barto-collapse}) shows that the (2,3)-minimality algorithm is a uniform polynomial-time algorithm
for the class $\cC$ of structures $\bbH$ such that $\CSP(\bbH)$ has bounded width. 
\end{example}

\begin{definition}
Let $\cM$ be a strong linear Maltsev condition.
\begin{itemize}

\item A \emph{uniform polynomial-time algorithm}
for $\cM$ is a uniform polynomial-time algorithm
for the class of structures satisfying $\cM$.

\item 
A \emph{semiuniform polynomial-time algorithm}
for $\cM$ is
a polynomial-time algorithm that,
when given as input a CSP instance $(\bbG, \bbH)$
and polymorphisms $f_1, \ldots, f_m$ of $\bbH$ that satisfy $\cM$,
correctly decides the CSP instance.
\end{itemize}
\end{definition}

\begin{example}
\label{ex:maltsev-op}
Let $\cM$ be the strong linear Maltsev condition
$\{ f(y,y,x) = x, f(x,y,y) = x \}$,
which asserts that $f$ is a Maltsev operation.
The algorithm due to~\cite{BulatovDalmau06-maltsev} is
readily verified to be
a semiuniform polynomial-time algorithm for $\cM$.
\end{example}

\begin{definition}
Let $\cM$ be a strong linear Maltsev condition.
\begin{itemize}

\item Define the \emph{metaquestion for $\cM$}
to be the problem of deciding, given a structure $\bbH$,
whether or not $\bbH$ satisfies $\cM$.

\item Define the \emph{creation-metaquestion for $\cM$}
to be the problem
where the input is a structure $\bbH$,
and the output is 
a sequence $f_1, \ldots, f_m$ of polymorphisms of $\bbH$
that satisfy $\cM$ in the case that $\bbH$ satisfies $\cM$,
and ``no'' otherwise.

\end{itemize}
\end{definition}

We note the following observation.

\begin{prop}
\label{prop:creation-metaquestion}
Let $\cM$ be a strong linear Maltsev condition.
If the creation-metaquestion for $\cM$ is polynomial-time computable,
then so is the metaquestion for $\cM$.
\end{prop}

The following theorem is the main theorem of this section.
It connects the existence of a uniform polynomial-time algorithm
for a Maltsev condition to
the existence of a semiuniform polynomial-time algorithm
and the tractability of the creation-metaquestion.
We say that a strong linear Maltsev condition $\cM$
is \emph{idempotent} if if it contains 
or entails
the identity $f(x,\ldots,x) = x$
for each operation symbol $f$ appearing in $\cM$.

\begin{theorem}
\label{thm:metaquestion-and-uniformity}
Let $\cM$ be a strong linear Maltsev condition.
\begin{itemize}

\item If the creation-metaquestion for $\cM$ is polynomial-time computable
and $\cM$ has a semiuniform polynomial-time algorithm,
then the condition $\cM$ has a uniform polynomial-time algorithm.

\item When $\cM$ is idempotent, 
the converse of the previous statement holds.

\end{itemize}
\end{theorem}

\begin{proof}
For the first claim, the polynomial-time algorithm is this.
Given a pair $(\bbG, \bbH)$, invoke the algorithm
for the creation-metaquestion; if a sequence
$f_1, \ldots, f_m$ of polymorphisms is returned,
then invoke the semiuniform polynomial-time algorithm
on the pair $(\bbG, \bbH)$ and the polymorphisms
$f_1, \ldots, f_m$, and output the result.

We now prove the second claim.  Assume that $\cM$
is idempotent and has a uniform polynomial-time algorithm.
It follows immediately that $\cM$ has a semiuniform
polynomial-time algorithm.
The algorithm for the creation-metaquestion (for $\cM$)
is the algorithm of 
Lemma~\ref{bwalgo}, but where instead of invoking the
$(2,3)$-consistency algorithm, one invokes the 
uniform polynomial-time algorithm for $\cM$.
\end{proof}

\begin{example} \label{ex:maltsev-equivalence}
Consider, as an example,
the condition $\cM$ from Example~\ref{ex:maltsev-op}.
Theorem~\ref{thm:metaquestion-and-uniformity}
implies that, for this condition,
the existence of a polynomial-time algorithm
for the creation-metaquestion
is equivalent to the existence of a uniform polynomial-time algorithm.
\end{example}

The following corollary, which exhibits a hypothesis
under which a uniform polynomial-time algorithm 
immediately implies tractability of a metaquestion,
 is immediate from 
Theorem~\ref{thm:metaquestion-and-uniformity}
and Proposition~\ref{prop:creation-metaquestion}.

\begin{corollary}
\label{cor:metaquestion-uniformity}
Let $\cM$ be a strong linear Maltsev condition that is idempotent.
If $\cM$ has a uniform polynomial-time algorithm,
then the metaquestion for $\cM$ is polynomial-time computable.
\end{corollary}

\section{Positive complexity results}

\subsection{Set polymorphisms}

We begin by studying set polymorphisms.

\begin{definition} 
\label{def:power-structure}
Let $\bbH$ be a relational structure. Define on $\cP(H)\setminus \emptyset$ a relational structure,
the \emph{power structure}  $\cP(\bbH)$,
of the same type as follows: if $\theta$ is a basic relation of $\bbH$ of arity $k$, declare $(X_1,\dots,X_ k) \in \theta'$ if, for every $1 \leq i \leq k$ and every $a \in X_i$, there exists a tuple $(x_1,\dots,x_k) \in \theta \cap \prod_{i=1}^k X_i$ such that $x_i = a$. A homomorphism from this structure to $\bbH$ is called a {\em set polymorphism of $\bbH$}. We'll say a set polymorphism $f$ is idempotent if $f(\{x\}) = x$ for all $x \in H$.\end{definition}


The presence of a set polymorphism admits multiple characterizations.

\begin{lemma} \label{lemma-width1}
\cite{DalmauPearson99-width1,FederVardi99-structure}
Let $\bbH$ be a relational structure. Then, the following are equivalent:
\begin{enumerate}
\item $\bbH$ has a set polymorphism;
\item there is a map $f:\cP(H)\setminus \emptyset \rightarrow H$ such that, for every $m \geq 1$, the operation defined by 
$f_m(x_1,\dots,x_m) = f(\{x_1,\dots,x_m\})$ is a polymorphism of $\bbH$.
\item $\bbH$ admits TS polymorphisms of all arities. 
\item $\CSP(\bbH)$ has width 1. 
\end{enumerate} \end{lemma}

We now observe that detecting a set polymorphism can be performed
in EXPTIME.  Note that this improves the naive complexity upper bound
of NEXPTIME, which is obtained by constructing
the structure defined in Definition~\ref{def:power-structure},
nondeterministically guessing a map to the given structure $\bbH$,
and then checking if the map is a homomorphism.

\begin{prop} 
\label{prop:set-polymorphism-exptime}
Deciding if a relational structure admits a set polymorphism is in EXPTIME. \end{prop}

\begin{proof} Notice that a structure $\bbH$ has a set polymorphism if and only if its core has an idempotent set polymorphism (see also the proof of Lemma \ref{lemma-down-to-core}). Indeed, let $\bbR$ denote the core of $\bbH$, with $R \subseteq H$, and let $r$ be a retraction of $\bbH$ onto $\bbR$, i.e. $r$ is an onto homomorphism and $r(x) = x$ for all $x \in R$. If $f$ is a set polymorphism of $\bbH$, then it is easy to see that the restriction $g$ of $r \circ f$ to subsets of $R$ is a set polymorphism for $\bbR$. Notice also that the one element sets $\{x\}$ with $x \in R$ induce an isomorphic copy of $\bbR$ in $\cP(\bbR)$; since $\bbR$ is a core the restriction of $g$ to this substructure induces an automorphism $\sigma$ of $\bbR$; then $\sigma^{-1} \circ g$ is an idempotent set polymorphism of $\bbR$. Conversely, if $g$ is a set polymorphism of $\bbR$, then define a set function $f$ on $H$ by $f(X) =g \left( \{r(x):x \in X\} \right)$. It is straightforward to verify that $f$ is a set polymorphism for $\bbH$. 

So now we proceed as follows: we first find the core of $\bbH$.  Loop over all mappings $f: H \rightarrow H$, but in a fashion that increases the cardinality of the size of the image of $f$. Check each mapping for being an endomorphism; once an endomorphism is found, the image of the original structure under the endomorphism is a core. This computation can be done in PSPACE and hence in EXPTIME.

Now that we have the core $\bbR$ of $\bbH$, we can test, in EXPTIME, if it admits an idempotent set polymorphism, in  the  manner of
 Lemma~\ref{bwalgo}: indeed, if the set polymorphism exists, and since $\bbR$ is a core,  then $\CSP(\widehat{\bbR})$ has bounded width; using the power structure $\cP(\bbR)$ as our instance, and fixing values one at a time, we shall either reject or eventually obtain a candidate set function which we can test for being a polymorphism. 
\end{proof}

As seen (Lemma~\ref{lemma-width1}), detecting for a set polymorphism
is equivalent to checking for the presence of TS polymorphisms of all arities.  In the quest to improve the complexity upper bound just given,
a natural question that one might ask
 is whether or not it suffices to check
for TS polymorphisms up to some bounded arity,
in order to ensure TS polymorphisms of all arities.
The following proposition answers this question in the negative.

\begin{prop}
 For every prime $p \geq 3$, there exists a digraph $\bbH$ with $p$ vertices that admits TSI polymorphisms of all arities strictly less than $p$ but no TS (in fact, no cyclic) polymorphism of arity $p$. 
 \end{prop}

\begin{proof}
Consider 
the directed cycle of length $p$, 
i.e. vertices $\{0,1,\dots,p-1\}$ and arcs $(i,i+1)$ for all $i$ (modulo $p$). 
Since this digraph has no loop, it cannot admit a cyclic polymorphism of arity $p$ because 
there is an arc from $f(0,1,\cdots,p-1)$ to $f(1,2,\cdots,p-1,0)$. 
On the other hand, if $2 \leq k < p$, define an operation $f$ as follows: given a tuple $(x_1,\dots,x_k)$, let $x_{i_1},\dots,x_{i_t}$ be the distinct representatives of the set $\{x_1,\dots,x_k\}$. Since $1 \leq t < p$, it admits a multiplicative inverse $t^{-1}$ modulo $p$. Then set
$$f(x_1,\cdots,x_k) = t^{-1}(x_{i_1}+\cdots + x_{i_t})$$ where the sum  is modulo $p$. It is straightforward to verify that this is a TSI polymorphism of the cycle. 
\end{proof}

\subsection{Conservative commutative polymorphisms}
\label{subsect:cc-polymorphisms}

We now prove an NL upper bound on the detection of
conservative commutative polymorphisms.  
Note that detecting such polymorphisms can be performed in 
polynomial time, by
Corollary~\ref{cor:even-conservative}.
However, the algorithm thusly given itself relies on the fact
that a CSP having a commutative conservative polymorphism
is decidable in polynomial-time; 
note that such a CSP can be complete for polynomial time
(this occurs even in the case of the polymorphisms $\wedge$ and $\vee$ on the domain $\{ 0, 1 \}$, see for 
example~\cite{ABISV09-refining}).
We believe that the present result is thus interesting
as it improves this polynomial-time upper bound.

\begin{theorem} 
\label{thm:cc-nl}
Deciding if a structure admits a conservative binary commutative polymorphism is solvable in non-deterministic logspace. \end{theorem}

In the scope of this proof, we refer to a
conservative binary commutative function as a \emph{cc-operation}.
For a binary operation $f: D^2 \to D$, 
and a $2$-element subset $C = \{ c, c' \}$ of $D$,
we will use the notation 
$f(C) = d$ to indicate that
$f(c,c') = f(c',c) = d$; 
also, if $f$ is a cc-operation, we will use
$f(C)$ to denote the value $f(c,c') = f(c',c)$.

\begin{proof}
Throughout, $S$ will denote a relation of the structure,
and $D$ will denote the domain of the structure.
For each relation $S$ of the structure and
each pair of tuples $s, s' \in S$,
define $S^{s,s'}$ to be the relation
$\{ t \in S ~|~ t_i \in \{ s_i, s'_i \} \textup{ for all $i$ } \}$.
We claim that for a relation $S$ of the structure,
a cc-operation $f$ preserves $S$ if and only if,
for each $s, s' \in S$,
it preserves $S^{s,s'}$.
For the forward direction, let $r, r' \in S^{s,s'}$.
It holds that $f(r,r') \in S$.  
But since $r_i, r'_i \in \{ s_i, s'_i \}$ for all $i$,
it holds by the conservativity of $f$ that
$f(r_i, r'_i) \in \{ s_i, s'_i \}$ for all $i$,
and so $f(r,r') \in S^{s,s'}$.
For the backward direction,
suppose that $r, r' \in S$.  It holds, by assumption,
that $f(r,r') \in S^{r,r'}$.
But since $S^{r,r'} \subseteq S$, we have $f(r,r') \in S$.

Now define $S^{s,s',-}$ to be the relation
obtained from $S^{s,s'}$ by projecting out
each coordinate $i$ such that the projection
$\pi_i(S^{s,s'})$ contains just one element.
Let $s^-, s'^{-}$ be the tuples in 
$S^{s,s',-}$ corresponding to $s$ and $s'$, respectively;
so, for each $i$ it holds that
$\pi_i(S^{s,s',-}) = \{ s^-_i, s'^{-}_i \}$.
As any cc-operation is idempotent,
we have that a cc-operation preserves 
all of the relations $S^{s,s'}$ if and only if
it preserves all of the relations $S^{s,s',-}$.
By the claim of the previous paragraph, we obtain:

{\bf Claim 1:}
A cc-operation preserves the relations of the structure
if and only if it preserves all of the relations
$S^{s,s',-}$.

For each relation $S^{s,s',-}$ and each tuple $r \in S^{s,s',-}$,
define $g_r: D^2 \to D$ to be the partial operation where
$g_r(\pi_i(S^{s,s',-})) = r_i$ for each $i$,
and $g_r(d,d) = d$ for each $d \in D$.
Each cc-operation $f$ that preserves $S^{s,s',-}$
must map the tuples $s^-, s'^{-}$ to a tuple $r \in S^{s,s',-}$;
hence, we obtain that such an $f$
must be an extension of $g_r$, for some 
tuple $r \in S^{s,s',-}$.
But observe that when such an $f$ extends such a $g_r$,
it holds that the partial operation $g_r$ itself preserves
$S^{s,s',-}$.
Define $[S^{s,s',-}]$ to be the subset of 
$S^{s,s',-}$ which contains a tuple $r \in S^{s,s',-}$
if and only if $g_r$ preserves $S^{s,s',-}$.
We thus obtain the following.

{\bf Claim 2:} A cc-operation $f$ preserves $S^{s,s',-}$
if and only if there exists a tuple $r \in [S^{s,s',-}]$
such that $f$ extends $g_r$.

Suppose, for the remainder of the proof,
that the domain $D$ of the structure 
is a subset of the integers.
Consider a relation $[S^{s,s',-}]$ of arity $k$.
Each tuple $r \in [S^{s,s',-}]$ can be naturally mapped
to a tuple $r^*$ in $\{ n, x \}^k$, where 
$r^*_i$ is equal to $n$ or $x$ depending on whether or not
$r_i$ is the mi{\bf n} or ma{\bf x} of $\pi_i(S^{s,s',-})$.
Define 
$[S^{s,s',-}]^* = \{ r^* ~|~ r \in S^{s,s',-} \}$.
In an analogous fashion, 
each cc-operation $f: D^2 \to D$ can be naturally viewed
as an operation $f^*: {D \choose 2} \to \{ n, x \}$, 
defined, for each $C \in {D \choose 2}$,
 by $f^*(C) = n \textup{ or } x$ depending on whether 
$f(C) = \min(C) \textup{ or } f(C) = \max(C)$.

From this discussion and the first two claims,
we obtain the following.

{\bf Claim 3:}
A cc-operation $f$ preserves the structure
if and only if 
it holds that, for each relation $S^{s,s',-}$ (of arity $k$),
the tuple
$(f^*(\pi_1(S^{s,s',-})), \ldots, f^*(\pi_k(S^{s,s',-})))$ is in
$[S^{s,s',-}]^*$.

We now observe the following.

{\bf Claim 4: }
Each relation of the form $[S^{s,s',-}]^*$ is
closed under the majority operation $m$ on $\{ n, x \}$.

We prove this claim as follows.
Set $(C_1, \ldots, C_k) =
(\pi_1(S^{s,s',-}), \ldots, \pi_k(S^{s,s',-}))$.
Let $t^1, t^2, t^3$ be arbitrary tuples in
$[S^{s,s',-}]^*$, 
and let $r^1, r^2, r^3$ be the corresponding
tuples in $[S^{s,s',-}]$.
Set $t = m(t^1, t^2, t^3)$,
and let $r$ be the tuple over $D$ that corresponds to $t$.
We show that for any $C_j$,
it holds that 
$g_r(C_j) = g_{r^3}(g_{r^1}(C_j),g_{r^2}(C_j))$,
which suffices (via the definition of $[S^{s,s',-}]$): 
$g_r$ is then obtainable
by composing the $g_{r^i}$, implying that $g_r$ preserves $S^{s,s',-}$.
We verify the identity as follows.
\begin{itemize}

\item If $t^1_j = t^2_j$, then 
$t_j = t^1_j = t^2_j$ and
$g_r(C_j) =  g_{r^1}(C_j) = g_{r^2}(C_j)$; thus
this last value is equal to 
$g_{r^3}(g_{r^1}(C_j),g_{r^2}(C_j))$,
by the idempotence of $g_{r^3}$.

\item If $t^1_j \neq t^2_j$, 
then $t_j = t^3_j$.
We have $\{ g_{r^1}(C_j),g_{r^2}(C_j)) \} = C_j$,
and so 
$g_{r^3}(g_{r^1}(C_j),g_{r^2}(C_j)) = g_r^3(C_j)$.
\end{itemize}

To determine whether or not 
there exists a cc-operation $f$ that preserves the original structure,
by Claim 3, one may determine whether or not the following
CSP instance has a solution.
The variable set is  $\{ X_C ~|~ C \in {D \choose 2} \}$,
where the variable $X_C$ represents the value $f^*(C)$;
for each relation of the form $S^{s,s',-}$ (of arity $k$),
there is a constraint stating that
the variable tuple 
$(X_{\pi_1(S^{s,s',-})}, \ldots, X_{\pi_k(S^{s,s',-})})$
must be mapped to a value in $[S^{s,s',-}]^*$.
By Claim 4, each of these relations is preserved by
 the majority operation $m$,
 and it is known that one can decide in NL all
 CSP instances on a two-element domain
 with a majority polymorphism~\cite{ABISV09-refining}.

Hence, it suffices to argue that each of the relations
$[S^{s,s',-}]^*$ can be computed in logspace.
To compute these relations, the algorithm loops over
each relation $S$ and each pair $s, s' \in S$.
For each tuple $t \in S$, it checks to see if
the tuple $t^- \in S^{s,s',-}$ corresponding to $t$
has the property that $g_{t^-}$ preserves $S^{s,s',-}$;
in the affirmative case, it holds that $t^-$ is in 
$[S^{s,s',-}]$, and so the algorithm outputs $(t^-)^*$.
The check  can be carried out in logspace as follows.  
Let $k$ denote the arity of $S$.
The algorithm loops over all pairs of tuples $u, u'$ in 
$S^{s,s'}$; for each such pair, the algorithm 
loops on $i = 1, \ldots, k$; for each value of $i$,
it attempts to find a tuple in $S^{s,s'}$
that is equal to $g_{t^-}(u,u')$ on the first $i$ coordinates,
and store a pointer to such a tuple.
For a particular value of $i$, it can loop over all tuples
in $S^{s,s'}$ and, for those having the correct value in the $i$th
coordinate, it can check for correctness in the first $i-1$ coordinates
by comparing with the tuple obtained after the $i-1$th iteration.
\end{proof}


\section{Complexity hardness results} 

In this section we prove various hardness results for the existence of ``good'' polymorphisms on finite structures.

\subsection{A generic NP-hardness result for Maltsev conditions}
\label{subsect:generic-hardness}
We here give a rather general result that applies to a wide range of Maltsev conditions. 
In order to present the theorem statement, we  introduce the following definitions.

\begin{definition}  Let $\cM$ be a strong, linear Maltsev condition. 
\begin{itemize}
\item We say that $\cM$  is {\em non-trivial} if, on a domain with at least 2 elements, no tuple of projections can satisfy it. 

\item $\cM$ is of {\em height 1} if its identities all have height 1. 
\item We'll say $\cM$ is {\em consistent} if for every non-empty finite set $D$, there exist idempotent operations on $D$ that satisfy $\cM$. 
\end{itemize}
\end{definition}

The following is the statement of the main theorem proved in this subsection. 

\begin{theorem} 
\label{thm:np-completeness-maltsev-condition}
Let $\cM$ be a non-trivial, consistent, strong linear Maltsev condition of height 1. The problem of deciding if a relational structure satisfies $\cM$ is NP-complete (even when restricted to at most binary relational structures). 
\end{theorem}

In order to present some examples of Maltsev conditions
to which Theorem~\ref{thm:np-completeness-maltsev-condition}
applies, we introduce the following definition.

\begin{definition}  Let $\cM$ be a strong, linear Maltsev condition. Let $\cM_q$ be the strong, linear Maltsev condition obtained from $\cM$ by replacing every identity of the form $f(x_{i_1},\dots, x_{i_n}) \approx x_{i_j}$ by the identity $f(x_{i_1},\dots, x_{i_n}) \approx f(x_{i_j},\dots, x_{i_j})$. \end{definition}

 If $\cM$ is the Maltsev condition defining Siggers (Maltsev, $k$-NU, ...), we'll say that an operation satisfying $\cM_q$ is {\em quasi-Siggers (Maltsev, $k$-NU, ...)}. 

 Our reason for introducing 
 (for example) quasi-Maltsev operations is this.
As mentioned,
 it is known that a structure $\bbH$ and its core $\bbH'$
 each give rise to the same CSP, that is,
 $\CSP(\bbH) = CSP(\bbH')$.
 Hence, if the core $\bbH'$ has a polymorphism 
 known to imply tractability of $\CSP(\bbH')$,
 such as a Maltsev polymorphism,
it follows 
that the problem $\CSP(\bbH)$ is also tractable.
As the following lemma implies, 
the assertion that a structure has a quasi-Maltsev polymorphism 
characterizes precisely when its core has a Maltsev polymorphism,
and hence precisely characterizes the sufficient condition
for tractability just described.
We remark that the notion of 
quasi- versions of operations
has been previously considered in the literature,
in particular, within the context of 
infinite-domain constraint satisfaction;
 see for example~\cite{BodirskyChen07-oligomorphic,BodirskyChen09-qualitative}.

    It is probable that some version of the following lemma is implicit in the literature (see also \cite{BartoOprsalPinsker15-wonderland}). At any rate, its proof is straightforward.
    
\begin{lemma} \label{lemma-down-to-core} 
Let $\bbB$ be a relational structure and let $\cM$ be a strong linear Maltsev condition. 
The following are equivalent:
\begin{enumerate}
\item $\bbB$ satisfies $\cM_q$;
\item The core of $\bbB$ satisfies $\cM$;
\item The core of $\bbB$ satisfies $\cM$ via idempotent operations.
\end{enumerate}
 \end{lemma}

\begin{proof} Let $\bbC$ denote the core of $\bbB$ and let $r$ be a retraction of $\bbB$ onto $\bbC$. $(1) \implies (3)$:  Let $f_1,\dots, f_s$ be polymorphisms of $\bbB$ that satisfy $\cM_q$. For each $i=1,\dots,s$ let $g_i $ denote the restriction of  $r \circ f_i$ to $C$ and let $\sigma_i(c) = g_i(c,\dots,c)$ for all $c \in C$. Since $\bbC$ is a core, each $\sigma_i$ is an automorphism of $\bbC$; we claim that the polymorphisms $\sigma_i^{-1} \circ g_i$ of $\bbC$ are idempotent and satisfy $\cM$. The first claim is immediate. Let $f_i(x_{i_1},\dots,x_{i_n}) \approx f_j(x_{j_1},\dots,x_{j_m})$ be in $\cM$. Notice first that, if we fix $c \in C$ and set $x_k = c$ for all $k$, then 
$f_i(c,\dots,c) = f_j(c,\dots,c)$, so $\sigma_i = \sigma_j$. It follows that if $a_{i_k}, a_{j_l}$ are elements of $C$ then 
$$ \sigma_i^{-1} \circ g_i(a_{i_1},\dots,a_{i_n}) =  \sigma_i^{-1} \circ r \circ f_i(a_{i_1},\dots,a_{i_n}) =  \sigma_j^{-1} \circ r \circ  f_j(a_{j_1},\dots,a_{j_m}) =  \sigma_i^{-1} \circ g_j(a_{j_1},\dots,a_{j_m}).$$
Now consider  an identity of $\cM$ of the form $f_i(x_{i_1},\dots,x_{i_n}) \approx x_{i_j}$. Since the $f_i$ satisfy $\cM_q$,  if we set $x_{i_k} = a_{i_k} \in C$ we have that $$f_i(a_{i_1},\dots,a_{i_n}) = f_i(a_{i_j},\dots,a_{i_j})$$ and hence we obtain
$$ \sigma_i^{-1} \circ g_i(a_{i_1},\dots,a_{i_n}) =  \sigma_i^{-1} \circ r \circ f_i(a_{i_1},\dots,a_{i_n}) =  \sigma_i^{-1} \circ r \circ  f_i(a_{i_j},\dots,a_{i_j}) =  \sigma_i^{-1} \circ\sigma_i(a_{i_j})  = a_{i_j}.$$

$(3) \implies (2)$ is trivial. $(2) \implies (1)$:  If $g_1,\dots,g_s$ are polymorphisms of $\bbC$ satisfying $\cM$, for each $i=1,\dots,s$ let $f_i = g_i(r(x_1),\dots,r(x_n))$ (where $n$ is the arity of $g_i$). Clearly these are polymorphisms of $\bbB$, and it is immediate that they satisfy every identity in $\cM_q$ which belongs to $\cM$. Now consider an identity of $\cM_q$ of the form $f_i(x_{i_1},\dots,x_{i_n}) \approx f_i(x_{i_j},\cdots,x_{i_j})$,  obtained from the identity $f_i(x_{i_1},\dots,x_{i_n}) \approx x_{i_j}$ in $\cM$. Since $g_i$ satisfies $\cM$ we get in particular that $g_i(x,x,\dots,x) = x$ for all $x \in C$ and thus, if we set $x_{i_k} = a_{i_k} \in B$, we obtain
$$ f_i(a_{i_1},\dots,a_{i_n}) = g_i(r(a_{i_1}),\dots,r(a_{i_n})) = r(a_{i_j} )= g_i(r(a_{i_j}),\cdots,r(a_{i_j})) = f_i(a_{i_j},\cdots,a_{i_j}).$$

\end{proof}

Define $\cMBW$ to be the strong linear Maltsev condition
that asserts that $v$ and $w$ are BW operations
(recall Definition~\ref{def:bw-operations}).

\begin{example} \label{example}
The following strong linear Maltsev conditions are  non-trivial, consistent and of height 1 (and hence satisfy the hypothesis of 
Theorem~\ref{thm:np-completeness-maltsev-condition}): 
\begin{enumerate} 
\item $k$-cyclic, 
\item $k$-symmetric, 
\item $k$-TS,
\item Quasi-$k$-NU,
\item Quasi-Maltsev,
\item Quasi-Siggers, 
\item  The condition $\cMBW$; 
recall that idempotent polymorphisms satisfying this condition characterise cores whose CSP has bounded width 
(Theorem~\ref{thm:bw-cores})

   \end{enumerate} \end{example}

The main technical tool needed for the proof 
of
Theorem~\ref{thm:np-completeness-maltsev-condition}
is the following lemma.

\begin{lemma} \label{lemma-all-or-nothing} Let $\bbG$ be a connected, undirected graph without loops and with at least 2 vertices. Then there exists a finite, binary relational structure $\bbB$, first order definable from $\bbG$, with the following properties: 
\begin{enumerate}
\item $\bbG$ admits a 3-colouring if and only if  $\bbB$ is not a core;
\item if $\bbG$ admits a 3-colouring, the core $\bbC$ of $\bbB$ satisfies the following: every idempotent operation on $C$ is a polymorphism of $\bbC$;
\item  there exists a 3-element subset $S$ of $B$ such that the restriction of any idempotent polymorphism $f$ of $\bbB$ to $S$  is a projection.
\end{enumerate}

\end{lemma}

\begin{proof} The universe of our structure is $B = V(\bbG) \times \{1,2,3\}$; for convenience we denote $(u,i)$ by $u_i$. Choose an arbitrary orientation of the edges of $\bbG$, and for each arc $e=(u,v)$, let $R_e$ be the following binary relation on $B$:
$$R_e = \{(u_i,v_j): i \neq j\}.$$
Then $\bbB = \langle B; R_e \, (e \in E(\bbG) \rangle$. It is easy to see that $\bbB$ if first-order definable from $\bbG$. \\

\begin{figure}[htb]
\begin{center}
\includegraphics[scale=0.6]{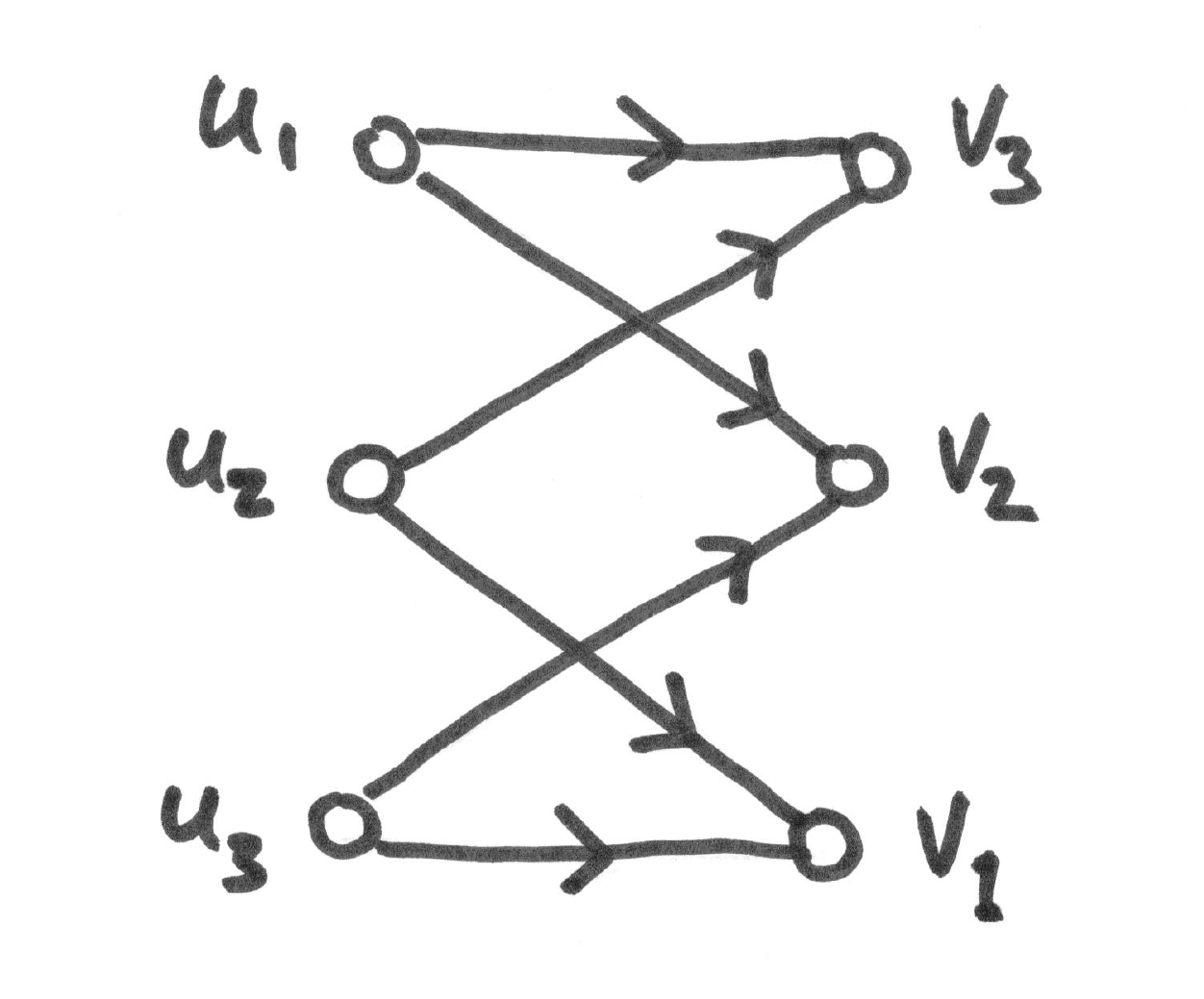}
 \caption{The relation $R_e$ associated to the arc $e=(u,v)$.} \label{relationRe}
 \end{center}
\end{figure}

{\bf Claim 1.} {\em If $\bbG$ admits a 3-colouring then $\bbB$ is not a core.}
 Let $\phi$ be a 3-colouring of $\bbG$. Consider the  self-map $r$ of $B$ defined by $r(u_i) = u_{\phi(u)}$ for all $1 \leq i \leq 3$. Then  $r$ preserves every $R_e$, since, if $e=(u,v)$, then $(r(u_i),r(v_j)) = (u_{\phi(u)}, v_{\phi(v)}) \in R_e$. \\

Notice that the map $r$ above is a proper retraction of $\bbB$ onto the substructure $\bbC$ induced by $C=\{u_{\phi(u)}: u \in V(\bbG)\}$.  This structure  is a (rigid) core: indeed, for every $e=(u,v)$, the relation $R_e$ restricted to $C^2$ is the singleton $\{(u_{\phi(u)},v_{\phi(v)})\}$; since $\bbG$ is connected and has at least 2 vertices, it means that every unary polymorphism of $\bbC$ must fix every element of $C$. \\

{\bf Claim 2.} {\em If $\bbB$ is not a core then  $\bbG$ admits a 3-colouring.}
 Let's suppose that $\bbB$ is not a core, and hence admits some unary, non injective polymorphism $f$. Notice that for each $u \in V(\bbG)$, the set $\{u_1,u_2,u_3\}$ is preserved by the polymorphisms of $\bbB$, as it is the projection on some coordinate of any $R_e$ such that $u$ is incident to $e$.  Since the sets $ \{u_1,u_2,u_3\}$ are pairwise disjoint,   the image of $f$ on some set $W=\{w_1,w_2,w_3\}$ has size at most 2. \\

\noindent {\bf Fact 1.} {\em Let $e=(u,v)$ be an arc of $\bbG$, and let $\{i,j,k\} = \{1,2,3\}$. 
\begin{enumerate} \item If $u_i$ and $u_j$ are in the image of $f$, then so is $v_k$;
\item If $u_i$ is in the image of $f$, then $v_j$ or $v_k$ is in the image of $f$.
 \end{enumerate}} 

We prove the first statement, the second is similar. By hypothesis there exist distinct $s,t$ such that $f(u_s)=u_i$ and $f(u_t)=u_j$. Since there exists some $r$ such that $(u_s,v_r),(u_t,v_r)\in R_e$, we have $(u_i,f(v_r)), (u_j,f(v_r)) \in R_e$ so $f(v_r) = v_k$. \\

 For a vertex $u \in V(\bbG)$ let $F(u)$ denote the set of indices that appear in the image of $f$ on $\{u_1,u_2,u_3\}$. It is immediate from the previous fact that, if $e=(u,v)$ is an arc of $\bbG$ and the image of $f$  restricted to $\{v_1,v_2,v_3\}$ has at most 2 elements, then the same holds for the set $\{u_1,u_2,u_3\}$. By connectedness of $\bbG$ and the fact that the image of $f$ on at least one set $W$ has size at most 2, it follows that  $|F(u)| \leq 2$ for all $u \in V(\bbG)$. Define a map $\phi: V(\bbG) \rightarrow \{1,2,3\}$ as follows:
 \[ \phi(u)  = \left \{ \begin{array}{ll}
         1, &\mbox{ if } F(u) \in \{\{1\},\{1,3\}\}, \\
  2, &\mbox{ if } F(u) \in \{\{2\},\{1,2\}\}, \\
3, &\mbox{ if } F(u) \in \{\{3\},\{2,3\}\}.

                        \end{array} \right.
       \]

\noindent {\bf Fact 2.} {\em $\phi$ is a proper 3-colouring of $\bbG$.} \\

Let $e=(u,v)$ and let $\phi(u)=i$. Suppose first that $F(u)=\{i\}$. Then by Fact 1 (2)  $F(v)$ cannot contain $i$. Inspection of the definition of $\phi$ shows that in this case $\phi(v) \neq i$. If on the other hand $F(u) = \{i,j\}$, then by Fact 1 (1) $F(v)$ must contain the unique $k \not\in \{i,j\}$; it is then immediate by the definition of $\phi$ that $\phi(v) \neq \phi(u)$. \\

It follows  from the preceding proofs that there are only two possibilities: either $\bbG$ is 3-colourable and the core of $\bbB$ is $\bbC$, or otherwise $\bbB$ itself is a core. To finish the proof of the theorem, it thus suffices to prove the following facts:\\

\noindent {\bf Fact 3.} {\em Assume $\bbG$ is 3-colourable. Then every idempotent operation on $C$ is a polymorphism of $\bbC$.} 

This is obvious: each basic relation of $\bbC$ has a single tuple. \\

\noindent {\bf Fact 4.} {\em The restriction of any idempotent polymorphism of $\bbB$ to a set $\{u_1,u_2,u_3\}$ is a projection.} 

(This construction is essentially due to Feder, \cite{Feder01classification}) Suppose without loss of generality that $e=(u,v)$ is an arc (otherwise we just reverse all the arcs in the gadgets). First define the relation $\theta=\{(u_i,v_i): i=1,2,3\}$ by the gadget  in Figure \ref{gadget}, where the arcs in the gadget denote the relation $R_e$. Since our pp-definition uses only constants and $R_e$, the idempotent polymorphisms of $\bbB$ preserve $\theta$. Then define the relation $\{(s,t): \exists z \, (s,z) \in \theta, (t,z) \in R_e\}$.   This pp-defines the complete graph on $\{u_1,u_2,u_3\}$, which is well-known to support no idempotent polymorphisms other than projections.  
\end{proof}

\begin{figure}[htb]
\begin{center}
\includegraphics[scale=0.6]{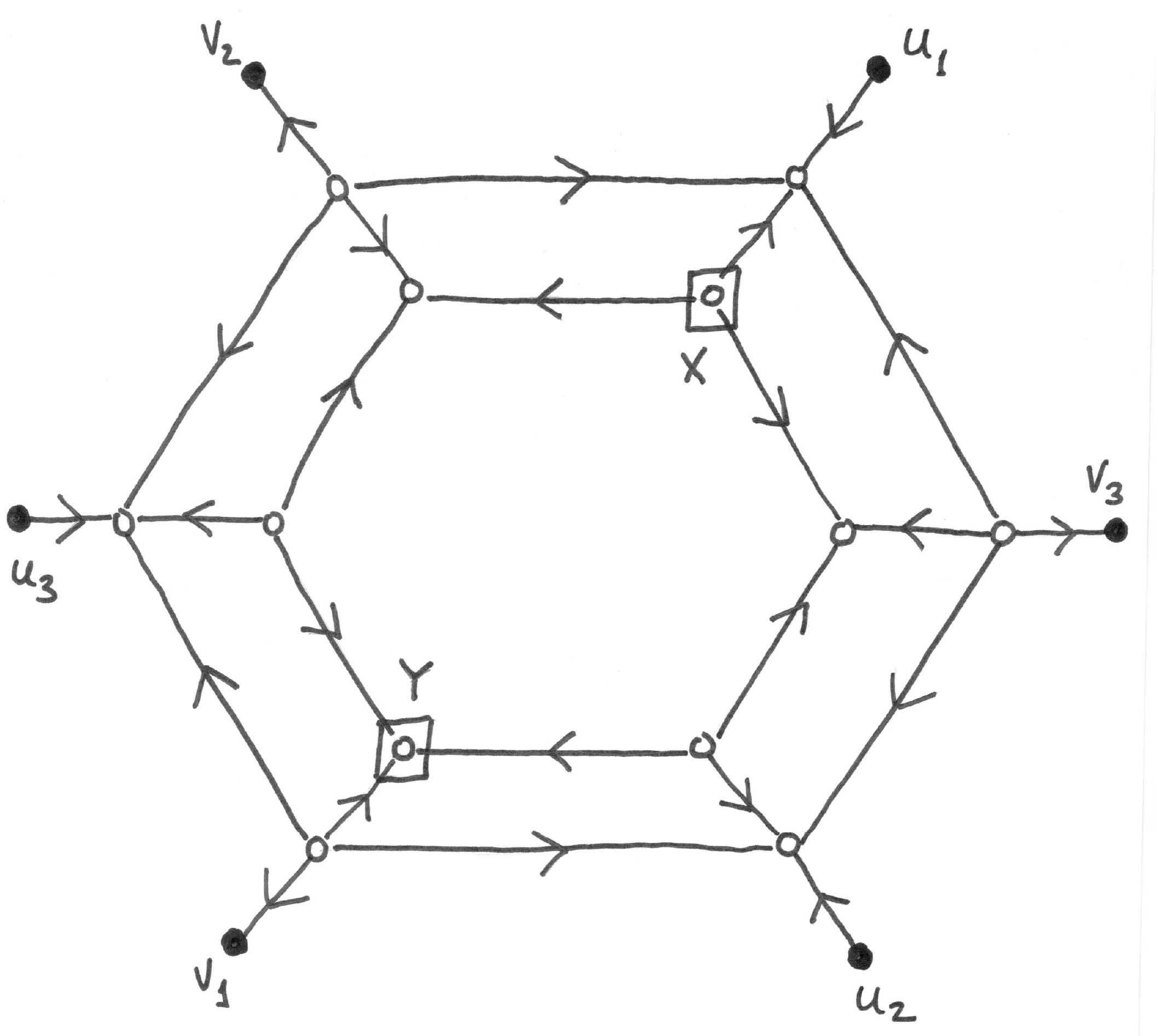}
 \caption{The gadget defining the relation $\theta$ in the proof of Fact 4 of Lemma \ref{lemma-all-or-nothing}.} \label{gadget}
 \end{center}
\end{figure}



\begin{proof} (Proof of Theorem~\ref{thm:np-completeness-maltsev-condition})
We reduce from 3-colourability, restricted to the types of graphs considered in Lemma \ref{lemma-all-or-nothing}. Consider the structure $\bbB$  associated to the  graph  $\bbG$ in Lemma \ref{lemma-all-or-nothing}: we prove that $\bbB$ satisfies $\cM$ if and only if $\bbG$ is 3-colourable.  
We first make the trivial observation that  $\cM_q = \cM$ since $\cM$ has height 1. Suppose first that $\bbB$ satisfies $\cM$. By Lemma \ref{lemma-down-to-core}, its core satisfies $\cM$ with idempotent polymorphisms. Since $\cM$ is non-trivial, the idempotent operations that witness it on the core of $\bbB$ cannot all restrict to projections on a set with 2 or more elements;  thus by Lemma  \ref{lemma-all-or-nothing} (3)  $\bbB$ is not a core so $\bbG$ is 3-colourable. Conversely, if  $\bbB$ does not satisfy $\cM$, then by Lemma \ref{lemma-down-to-core}  its core does not satisfy $\cM$ with idempotent polymorphisms; since $\cM$ is consistent, it follows from Lemma  \ref{lemma-all-or-nothing} (2) that the core of $\bbB$ is $\bbB$ itself,  and hence $\bbG$ is not 3-colourable. Thus the problem is NP-hard.

To see that the problem is in NP: it suffices to guess the polymorphisms that satisfy $\cM$.   \end{proof}

We now turn to consider the complexity of deciding if a structure has bounded width.
We first note the following fact, which is now folklore
(we present it here, as we do not know of an explicit reference).

\begin{prop}
\label{prop:bw}
A structure $\bbH$ has bounded width if and only if it satisfies
$\cMBW$.
\end{prop}

\begin{proof}
Let $\cM$ be the strong linear Maltsev condition obtained
by taking $\cMBW$ and adding the requirements that 
$v$ and $w$ are idempotent.
We have that a structure satisfies $\cM_q$ if and only if
it satisfies $\cMBW$.
Hence, 
by Lemma~\ref{lemma-down-to-core} 
a structure $\bbH$ satisfies $\cMBW$
if and only if the core of $\bbH$ satisfies $\cM$.
The core of $\bbH$ has bounded width if and only if it satisfies $\cM$,
by Theorem~\ref{thm:bw-cores}.
By Proposition~\ref{prop:bw-preserved-by-hom-equivalence},
$\bbH$ has bounded width if and only if its core does,
yielding the proposition.
\end{proof}

The following result settles the complexity of the just-mentioned decision problem, and answers a question of L. Barto 
(see \cite{Barto14-collapse}, just after Corollary 8.5).

\begin{corollary} 
\label{cor:deciding-bounded-width}
Deciding, given a structure $\bbH$, whether or not $\CSP(\bbH)$ has bounded width is NP-complete. \end{corollary} 

\begin{proof} 
This follows directly from Proposition~\ref{prop:bw} and
Theorem~\ref{thm:np-completeness-maltsev-condition}.\end{proof}

Recall that the algebraic dichotomy conjecture 
(Conjecture~\ref{conj:algebraic-dichotomy})
predicts that the relational structures $\bbH$ for which $\CSP(\bbH)$ is tractable are precisely those whose core admits a Siggers polymorphism; 
by Lemma \ref{lemma-down-to-core} those are precisely the structures that admit a quasi-Siggers polymorphism.
By the last theorem, this condition is also NP-complete.

\begin{corollary} 
\label{cor:deciding-quasi-siggers}
Deciding if a relational structure admits a quasi-Siggers polymorphism is NP-complete.  \end{corollary} \qed

The property of admitting $k$-symmetric polymorphisms of all arities $k \geq 2$ characterises structures whose CSP admits a special kind of approximation algorithm based on linear programming \cite{DBLP:conf/innovations/KunOTYZ12},  see also \cite{MR3146708} and \cite{KroCar-symmetric} Remark 5. For the purposes of the next result, we call a structure {\em all-$k$-symmetric} if it has this property. Notice that Lemma ~\ref{lemma-all-or-nothing} has consequences slightly stronger than those stated in Theorem~\ref{thm:np-completeness-maltsev-condition} as it shows NP-hardness of various Maltsev conditions (not necessarily strong ones):

\begin{corollary} 
\label{cor:deciding-ts-symm-cyc}
Let $k \geq 2$. Deciding if a relational structure admits any of the following polymorphisms is NP-complete:
\begin{enumerate}
\item a $k$-ary totally symmetric polymorphism,
\item a $k$-ary symmetric polymorphism,
\item a $k$-ary cyclic polymorphism.

\end{enumerate}
Furthermore, deciding the existence of a set polymorphism, or if a structure is \\all-$k$-symmetric is NP-hard. 
\end{corollary}

\begin{proof} Statements (1)-(3) follow directly from Theorem \ref{thm:np-completeness-maltsev-condition} and  Example \ref{example}.  To prove that the problem of deciding if a structure is all-$k$-symmetric, we also invoke Example \ref{example} and then proceed as in the proof of Theorem  \ref{thm:np-completeness-maltsev-condition} simply noticing that if $\bbG$ is 3-colourable then the core of the associated structure $\bbB$ admits idempotent $k$-symmetric polymorphisms for all $k$, and if $\bbG$ is not 3-colourable, then the core of $\bbB$ is $\bbB$ itself, and admits no $k$-symmetric polymorphism for any $k$. For the last statement, we also proceed exactly as in the proof of Theorem  \ref{thm:np-completeness-maltsev-condition}: by the proof of Proposition \ref{prop:set-polymorphism-exptime}, the structure $\bbB$ associated to the  graph  $\bbG$ in Lemma \ref{lemma-all-or-nothing} admits a set polymorphism precisely when its core admits an idempotent set polymorphism. It is easy to see that the TS polymorphisms guaranteed by Lemma \ref{lemma-width1} are idempotent if the core admits an idempotent set polymorphism; since no TS operation can be a projection when restricted 
to a non-trivial set,  by Lemma \ref{lemma-all-or-nothing},  if the core of $\bbB$ has an idempotent set polymorphism then  $\bbB$ is not a core and so $\bbG$ is 3-colourable. Conversely, since every non-trivial set admits an idempotent set function (take for instance the union operation), if the core of $\bbB$ does not admit an idempotent set polymorphism then $\bbB$ is not a core, and so $\bbG$ is 3-colourable. \end{proof}
\qed

\subsection{Semilattice polymorphisms}
\label{subsect:semilattice}


A known sufficient condition for a structure to have a set polymorphism
is that the structure has a 
semilattice polymorphism~\cite{DalmauPearson99-width1}.
We show that detecting this sufficient condition
(that is, for a semilattice polymorphism) is NP-hard,
even in the case where one restricts attention to conservative
polymorphisms (see \cite{DBLP:journals/ai/GreenC08} for related results.)

\begin{theorem} 
\label{thm:semilattice-hard}
Deciding if a relational structure admits any of the following  is NP-complete (even when restricted to at most binary relational structures): 

\begin{enumerate}
\item a  semilattice polymorphism,
\item a conservative semilattice polymorphism,
\item a commutative, associative polymorphism 
(that is, a commutative semigroup polymorphism).

\end{enumerate}

\end{theorem}

\begin{proof} Inclusion in NP holds, as one may guess the operation and verify that it has the desired form.
For NP-hardness, we use a reduction from 
the classical NP-complete problem
\emph{betweenness}~\cite{Opatrny79-betweenness}: 
\begin{itemize}
\item Input: A list of triples $(i,j,k)$ of distinct integers in $\{1,\dots,n\}$;
\item Question: is there a linear ordering of $\{1,\dots,n\}$ such that for each triple $(i,j,k)$ in the list, $j$ is between $i$ and $k$ ?
\end{itemize}

For a given list of triples, 
construct the following relational structure $\bbH$: 
its universe is $H = \{1,\dots,n\}$; 
each non-empty subset of $H$ 
having size less than or equal to $2$ is a relation; 
and, for each triple $(i,j,k)$ in the list,
we have the relation $\{(i,j),(j,k)\}$. 
We claim that the following are equivalent:
\begin{enumerate}
\item there exists a linear ordering satisfying the betweenness condition;
\item $\bbH$ admits a conservative semilattice polymorphism;
\item $\bbH$ admits a  semilattice polymorphism;
\item  $\bbH$ admits a commutative, associative polymorphism.
\end{enumerate}
For (1) $\Rightarrow$ (2),
if there is such an ordering, 
it is easily verified that the maximum operation of the ordering preserves every relation; this is a conservative semilattice polymorphism of $\bbH$.
The implications (2) $\Rightarrow$ (3) $\Rightarrow$ (4) 
are immediate.
For the implication (4) $\Rightarrow$ (1),
suppose $\bbH$ has a commutative, associative polymorphism $f$. Notice that we can induce a partial ordering of $\{1,\dots,n\}$ by setting, for distinct $u,v$,
$u < v \iff f(u,v) = v$. Indeed, the associative condition guarantees this relation is transitive. Choose any linear extension of this partial ordering; it is immediate that it satisfies the in-betweenness condition.  
\end{proof}

\subsection{An NL-hardness result for certain conservative polymorphisms}
\label{subsect:nl-hardness}
The remainder of this section focuses on the proof of the following theorem,
which establishes an NL-hardness result for a number of different
types of conservative polymorphisms, in the case of 
at most binary structures.

\begin{theorem} \label{thm-nl-hard} 
 Deciding if a relational structure admits 
any of the following polymorphisms is NL-hard, 
under first-order reductions,
even in the case of at most binary structures.
\begin{itemize}

\item a commutative conservative binary polymorphism,
\item a conservative $k$-ary TS polymorphism (for any $k \geq 2$),
\item a conservative set polymorphism,
\item a conservative $k$-ary symmetric polymorphism (for any even $k \geq 2$),
\item a conservative $k$-ary cyclic polymorphism (for any even $k \geq 2$).

\end{itemize}
\end{theorem}

Before launching into the proof of the result we require some preparation. Given a set $\cT$ of triples $(i,j,k)$ of distinct integers in $\{1,...,n\}$, let $\bbH(\cT)$ denote the relational structure defined as follows 
(note that this is a slight variant of the structures used in 
Theorem~\ref{thm:semilattice-hard} above):
its universe is $\{1,\dots,n\}$, and for each triple $(a,b,c) \in \cT$ we introduce a relation $\{(a,b),(b,b),(b,c)\}$. Furthermore, let $\bbD(\cT)$ denote the following digraph: 
its vertices are the pairs $(a,b)$ of distinct integers that appear consecutively in some triple of $\cT$, 
and we have an arc $(a,b)\rightarrow(b',c)$ 
iff $b=b'$ and either $(a,b,c) \in \cT$ or $(c,b,a) \in \cT$.  

\begin{lemma} \label{lemma} Let $\cT$ be a set of triples, and let $k \geq 2$. Then the following are equivalent:
\begin{enumerate}
\item  For all $a,b$, $(a,b)$ and $(b,a)$ lie in different strong components of $\bbD(\cT)$;
\item $\bbH(\cT)$  admits a  conservative commutative binary polymorphism;
\item $\bbH(\cT)$  admits a  conservative $k$-ary TS polymorphism;
\item  $\bbH(\cT)$ admits a  conservative set polymorphism.
\end{enumerate}

If $k \geq 2$ is even, then these conditions are equivalent to the following:

\begin{enumerate}
\item[(5)]  $\bbH(\cT)$  admits a  conservative $k$-ary symmetric polymorphism;
\item[(6)] $\bbH(\cT)$  admits a  conservative $k$-ary cyclic polymorphism.
\end{enumerate}

\end{lemma}

\begin{proof} Let $\bbD = \bbD(\cT)$ and $\bbH = \bbH(\cT)$. First notice that, since the projection of any basic relation of $\bbH$ on any one coordinate has size 2, the value of a polymorphism on tuples with at least 3 distinct entries is irrelevant; thus from a commutative binary polymorphism we can trivially build a TS polymorphism of any arity, and it follows that (2), (3) and (4) are equivalent. It is immediate that (3) implies (5) and (5) implies (6) for all $k$; finally if $k=2n$ then (6) implies (2): if $t$ is a conservative cyclic polymorphism of arity $2n$ then $f(x,y) = t(x,\cdots,x,y,\cdots,y)$ is a conservative, commutative polymorphism, where we identified the first $n$ variables and the last $n$ respectively. 

Suppose (2) holds and let $f$ denote the polymorphism witnessing this.  By definition of $\bbH$, it is easy to see that, if $(x,y)\rightarrow(u,v)$ in $\bbD$, then  $f(x,y) = x$ implies that $f(u,v)=u$.  Suppose for a contradiction that (1) does not hold; then there is a pair $(a,b)$ with a directed path to $(b,a)$, thus if $f(a,b) = a$ then $f(b,a)=b$, a contradiction. Hence $f(a,b) = b$; by commutativity we get $f(b,a)=b$; using the path from $(b,a)$ back to $(a,b)$ we force $f(a,b) =  a$, again a contradiction. 

Suppose that (1) holds. By our previous arguments it suffices to prove (2), and by definition of $\bbH$, to do this it suffices to define the polymorphism only on pairs that are vertices of $\bbD$ (the polymorphism can be defined arbitrarily on all other pairs of distinct elements).  Notice that, by definition, the map $(a,b) \mapsto (b,a)$ is an edge-reversing bijection of $\bbD$, i.e. $(a,b) \rightarrow (c,d)$ if and only if $(d,c) \rightarrow (b,a)$.
Let $F$ be a down set of $\bbD$ (i.e. $y \in F$ and $x\rightarrow y$ implies $x \in F$) which is maximal for the property $(a,b) \in F \Rightarrow (b,a) \not\in F$.  There is at least one such non-empty set: indeed, by hypothesis all strong components of $\bbD$ satisfy the property and at least one strong component is a down set (if a vertex does not lie in a directed cycle we consider it to be a strong component). We claim that every $(x,y) \not\in F$ satisfies $(y,x) \in F$. Indeed, if this is not the case, then there exists $(x,y) \not\in F$ with $(y,x) \not\in F$ and such that the strong component $S$ of $(x,y)$ is minimal  for this property, with respect to the ordering induced on strong components by $\bbH$. Let  $(u,v)$ admit a directed path to $(x,y)$; then we obtain a directed path from $(y,x)$ to $(v,u)$, and since $(y,x) \not\in F$ we conclude that $(v,u) \not\in F$. By minimality of $S$, it follows that  $(u,v)\in S$ or $(u,v) \in F$. Hence $F' = F \cup S$ is  a downset. By maximality of $F$, it follows there must exist some $(c,d) \in F$ with $(d,c) \in S$. But then $(d,c)$ admits a directed path to $(x,y)$ and $(c,d) \in F$, again a contradiction. 
By the just-proved claim and the definition of $F$, for each pair $(x,y)$ of $\bbD$, we have that exactly one of $(x,y),(y,x)$ is in $F$.
Define $f(x,y) = y$ for $(x,y) \in F$ and $f(x,y)=x$ otherwise. By our previous remark $f$ is commutative, and it is obviously conservative. Finally, we must show that $f$ is a polymorphism of $\bbH$. Fix a triple $(a,b,c) \in \cT$. It is easy to see that, since $f$ is conservative and commutative, it suffices to prove that $(f(a,b),f(b,c)) \in R = \{(a,b),(b,b),(b,c)\}$. The only bad case is if $f(a,b) = a$ and $f(b,c)=c$; by definition of $f$ this can only happen if $(b,c) \in F$ and $(a,b) \not\in F$, but since $(a,b) \rightarrow (b,c)$ this is impossible. 

\end{proof}

\begin{proof} ({\em of Theorem \ref{thm-nl-hard}}) 

Our goal is to reduce the following (or rather, the negation of this decision problem), via a first-order reduction,  to our problem of deciding the existence of a conservative 2-TS polymorphism:

\begin{itemize}
\item Input: a digraph $\bbK$ with two specified, distinct vertices $s$ and $t$;
\item Question: are  $s$ and $t$ in the same strong component of  $\bbK$ ?
\end{itemize}

This problem is easily seen to be NL-hard, as we can restrict it to digraphs with an arc from $t$ to $s$ to obtain the standard directed reachability problem (here we use the fact that $NL = co-NL$). Let $\bbK$ be our input digraph: we may  assume without loss of generality that $\bbK$ has no isolated vertices and no loops. It is clear that, if we construct a new digraph $\bbG$ by replacing every arc $e=(u,v)$ of $\bbK$ by a directed path of length 3, $e_1=(u,x), e_2=(x,y), e_3=(y,v)$, then there are directed paths from $s$ to $t$ and back in  $\bbG$ if and only if the same holds in $\bbK$; notice also that the construction is clearly first order.

Our overall strategy is the following: we will first exhibit a set of triples $\cT$ whose associated digraph $\bbD(\cT)$ is isomorphic to the disjoint union of  $\bbG$ and  $\overline{\bbG}$, this last digraph being obtained from $\bbG$ by flipping all its arcs.  
We can then throw in a few other triples to glue these copies together so that the resulting digraph $\bbD(\cT')$ contains pairs $(a,b)$ and $(b,a)$ that lie on a directed cycle if and only if $s$ and $t$ have the same property in $\bbG$. The construction of the set  $\cT$ is fairly straightforward,  it consists of the following triples: for each arc $\alpha = (u,v)$ of $\bbK$, we have the triples $(u_1,u_2,\alpha), (u_2,\alpha,v_1), (\alpha,v_1,v_2)$.  \\

\noindent{\bf Claim 1.} {\em The digraph $\bbD(\cT)$ is  isomorphic to the disjoint union of $\bbG$ and  $\overline{\bbG}$. }\\

Indeed, let $\Gamma$ be the set of pairs $(c,d)$ such that $(c,d,e)$ or $(e,c,d)$ is in $\cT$ for some $e$,  and let $\overline{\Gamma} = \{(y,x): (x,y) \in \Gamma\}$. It is clear that these sets are disjoint and their union equals the vertex set of $\bbD(\cT)$; furthermore no arc of $\bbD(\cT)$ connects vertices from $\Gamma$ and  $\overline{\Gamma}$. To complete the proof, it suffices to prove that the subdigraph of $\bbD(\cT)$  induced by $\Gamma$ is isomorphic to $\bbG$. Indeed, the pairs $(u_1,u_2)$ constitute a copy of the set $K \subset G$; given an arc $\alpha = (u,v)$ in $\bbK$, we have a path 
$e_1=(u,x), e_2=(x,y), e_3=(y,v)$ in $\bbG$, which corresponds exactly to the arcs induced by the triples associated to $\alpha$ with $x$ mapped to $(u_2,\alpha)$ and $y$ to $(\alpha,v_1)$; and obviously every vertex and arc of the subdigraph induced by $\Gamma$ is of this form. \\

We are now ready to define the structure we need (see Figure \ref{digraph}.)  Let $H$ be the disjoint union of the underlying set of the triples $\cT$ and $\{a,b\}$ where $a$ and $b$ are two new distinct elements.  Let $\cT'$ be the following set of triples on $H$:
$ \cT' = \cT \cup \cL $ where 
$$\cL =\left\{(a,b,s_1),(b,s_1,s_2), (s_1,s_2,a), (s_2,a,b),(t_1,a,b),(t_2,t_1,a), (b,t_2,t_1), (a,b,t_2)\right\}.$$
For convenience in the proof,  let $\overline{u} = (y,x)$ if $u = (x,y)$ is any vertex of $\bbD(\cT')$, and extend this notation to subgraphs in the obvious way. Let $L$ (for left) be the set of pairs $(x,y)$ such that $(x,y,e)$ or $(e,x,y)$ is in $\cL$ for some $e$, and let $R = \overline{L}$.
By the last claim, we may view $\bbD(\cT)$ as the disjoint union of $\bbG$ and $\overline{\bbG}$; notice that this is an induced subdigraph of $\bbD(\cT')$, since every triple of $\cL$ contains an occurrence of $a$ or $b$. \\

\noindent{\bf Claim 2.} {\em There exists a pair $(x,y)$ in the same strong component of  $\bbD(\cT')$ as  $(y,x)$ if and only if $s$ and $t$ are in the same strong component of $\bbG$.} \\

 The following observations are immediate:
  \begin{itemize} 
 \item[(i)] $s = (s_1,s_2)$, $ \overline{t} = (t_2,t_1)$ and $L$ lie in the same strong component of $\bbD(\cT')$; similarly $t$, $\overline{s}$ and $R$ lie in the same strong component of $\bbD(\cT')$;
\item[(ii)] any directed path connecting  $\bbG$ and $\overline{\bbG}$ must pass through  $s$ and $\overline{t}$ (i.e. through $L$) or through $t$ and $\overline{s}$ (i.e. through $R$).
\item[(iii)]  In $\bbD(\cT')$, there is a directed path from $u$ to $v$ if and only if there is a directed path from $\overline{v}$ to $\overline{u}$; if the first path lies in $\bbG$ ($\overline{\bbG}$) the the second lies in $\overline{\bbG}$ ($\bbG$ respectively).

\end{itemize}

Suppose there exists a pair $u$ such that $u$ and $\overline{u}$ are in the same strong component of $\bbD(\cT')$. By the preceding observations, we may safely assume that $u$ is a vertex of $\bbG$ (and thus $\overline{u}$ is a vertex of $\overline{\bbG}$.) There are, without loss of generality, only two cases to consider: 

(1) {\em The directed paths joining $u$ to $\overline{u}$ and back both go through $L$}: then $u$ and $s$ lie in the same strong component of $\bbG$ and $\overline{u}$ lies in the strong component of $\overline{t}$ in 
$\overline{\bbG}$. By observation (iii), we obtain that $u$ is in the strong component of $t$ in $\bbG$ and we are done. (2) {\em One directed path goes through $L$ and the other through $R$.}  If this is the case there is (without loss of generality) a directed path in $\bbG$ from $s$ to $t$, and a directed path in $\overline{\bbG}$ from $\overline{s}$ to $\overline{t}$; by (iii), we obtain a directed path in $\bbG$ from $t$ to $s$ and we are done. 

The converse is immediate: if $s$ and $t$ are in the same strong component of $\overline{\bbG}$ then by observation (i) $t = (t_1,t_2)$ and $ \overline{t} = (t_2,t_1)$ are in the same strong component of $\bbD(\cT')$. \\

 It follows from Claim 2 and Lemma \ref{lemma} that the structure $\bbH(\cT')$ admits no conservative, commutative binary polymorphism if and only if there are paths in $\bbG$ from $s$ to $t$ and from $t$ to $s$. To complete the proof, we briefly outline why the reduction is first order. It is very easy to define the triples of $\cT$ from the digraph $\bbK$, one may use the constants $s$ and $t$ to play the role of the indices 1 and 2. Finally,  rewriting the triples as a binary relation is obviously no problem.

\end{proof}

\begin{figure}
\colorbox{black}{\includegraphics[width=0.8\hsize,angle=90]{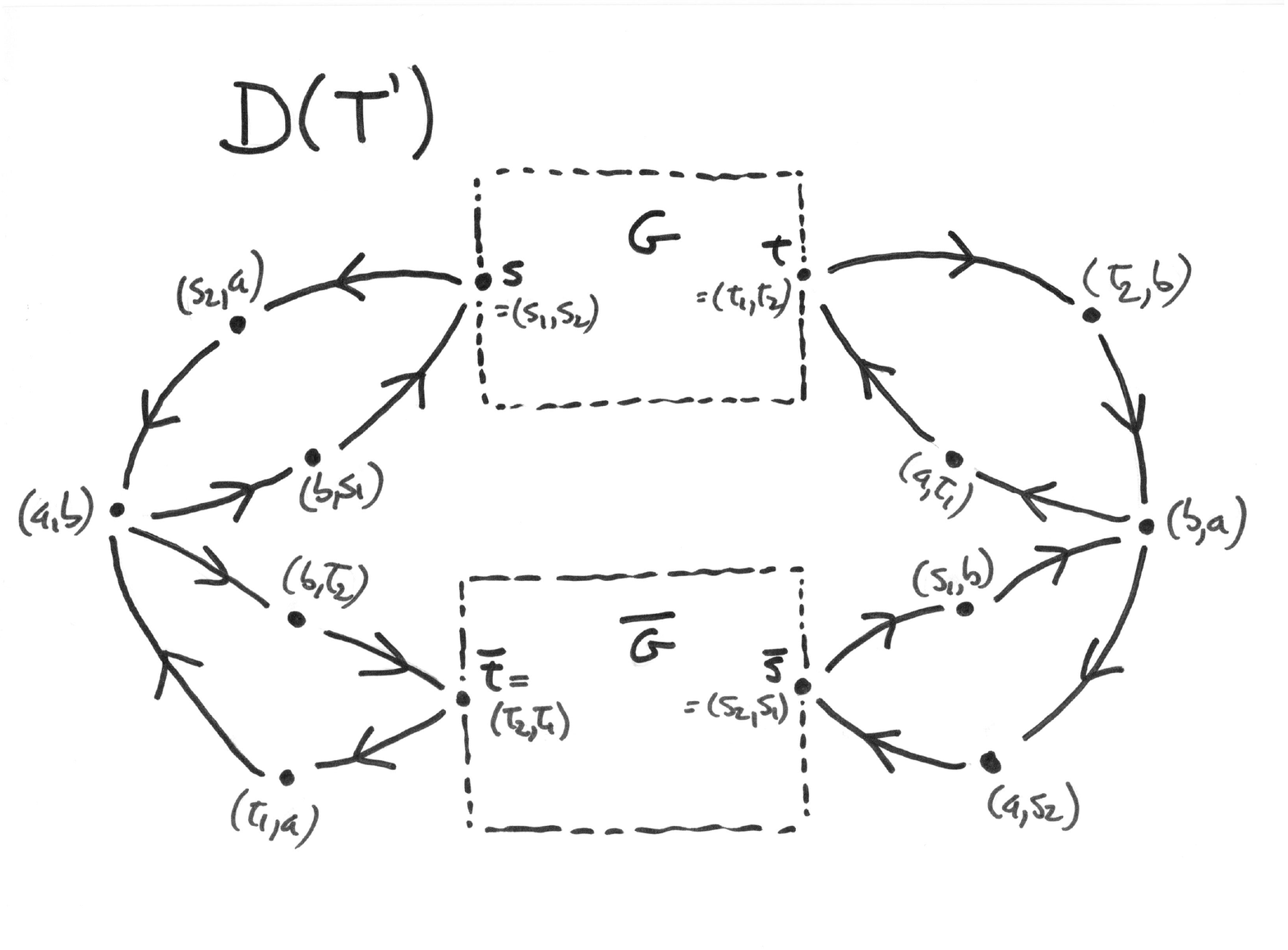}}
 \caption{The digraph $\bbD(\cT')$.} \label{digraph}
\end{figure}



\newpage

\section{Discussion}

In this section, we place in context
 some of the main themes of this article.

Let us first collect together some known relationships among the
polymorphisms and conditions that were considered.

Concerning idempotent polymorphisms, we have the following chain of implications.

\begin{center}
semilattice polymorphism $\Rightarrow$ idempotent set polymorphism
$\Rightarrow$ idempotent $k$-totally symmetric polymorphism
$\Rightarrow$ idempotent $k$-symmetric polymorphism
$\Rightarrow$ idempotent $k$-cyclic polymorphism
$\Rightarrow$ Siggers polymorphism
\end{center}

By $X \Rightarrow Y$, we mean to indicate that if a structure
has a polymorphism of type $X$, then it also has a polymorphism 
of type $Y$.  The first two implications here follow from
~\cite{DalmauPearson99-width1,FederVardi99-structure};
the 
next two are direct; and, the last 
follows from Proposition~\ref{prop:polymorphism-to-siggers}.

Concerning general polymorphisms, we have the following chain of implications.

\begin{center}
 set polymorphism
$\Rightarrow$  $k$-totally symmetric polymorphism
$\Rightarrow$  $k$-symmetric polymorphism
$\Rightarrow$  $k$-cyclic polymorphism
$\Rightarrow$ quasi-Siggers polymorphism
\end{center}

The first implication follows from ~\cite{FederVardi99-structure},
the next two are direct, and the last can be argued 
in the following way, via 
Lemma~\ref{lemma-down-to-core}: if a structure $\bbH$ has a 
$k$-cyclic polymorphism, then its core has an idempotent
$k$-cyclic polymorphism and hence a Siggers polymorphism;
it follows that $\bbH$ has a quasi-Siggers polymorphism.
Also, the two chains so far are related in the following way:
an idempotent set polymorphism is a particular type of set polymorphism,
so an idempotent set polymorphism directly implies a set polymorphism;
an analogous statement holds for $k$-totally symmetric polymorphisms,
$k$-symmetric polymorphisms, and $k$-cyclic polymorphisms.
It is readily verified that a Siggers polymorphism implies a
quasi-Siggers polymorphism.

As discussed earlier, \emph{bounded width} is a general known
sufficient condition for tractability of $\csp(\bbH)$.
It is known that if $\bbH$ has a $k$-totally symmetric polymorphism,
then $\bbH$ has bounded width; this is because, via
Lemma~\ref{lemma-down-to-core}, under the assumption,
the core of $\bbH$ has an idempotent $k$-totally symmetric polymorphism,
implying (by the proof of Corollary~\ref{cor:idempotent-ts-p})
that this core has bounded width, which in turn implies
that $\bbH$ has bounded width
(via Proposition~\ref{prop:bw-preserved-by-hom-equivalence}).

Another general known sufficient condition for the tractability
of $\csp(\bbH)$ is \emph{few subpowers}, studied by
\cite{BIMMVW10-fewsubpowers,IMMVW10-tractabilityfewsubpowers}.  
It is known that either the presence of a Maltsev polymorphism
or a $k$-near unanimity polymorphism implies 
few subpowers~\cite{BIMMVW10-fewsubpowers}
(although the metaquestion for $k$-near unanimity operations
is tractable, recall Theorem~\ref{cor:knu-polytime}).

We gave a general hardness result, 
Theorem~\ref{thm:np-completeness-maltsev-condition},
showing that, for a number of the quasi-versions of polymorphisms
and of the polymorphism types where idempotentcy is not required,
the metaquestion is hard 
(recall Example~\ref{example}
and
Corollary~\ref{cor:deciding-ts-symm-cyc}).

Concerning the applicability of the developed theory,
one can ask the following question.  
In the case that a structure $\bbH$ is known to 
have polymorphisms of some type that guarantee tractability
of $\csp(\bbH)$, in which cases can instances of $\csp(\bbH)$
be solved efficiently?
In the case of polymorphisms that imply
bounded width, the $(2,3)$-minimality algorithm can be employed
to efficiently solve the named instances
(recall Theorem~\ref{thm-barto-collapse}).
Otherwise, 
the notion of \emph{uniform polynomial-time algorithm}
from Section~\ref{sect:uniformity-and-metaquestions}
formalizes
the efficient solvability of $\csp(\bbH)$
over structures $\bbH$ having desirable polymorphisms of some type.
In the case of idempotent polymorphisms,
Theorem~\ref{thm:metaquestion-and-uniformity}
gives a characterization
of the existence of a uniform polynomial-time algorithm.





\section{Open issues}

We end by discussing a number of open issues
and posing further questions about metaquestions.

\begin{itemize}

\item A perusal of the table yields that, for some of the polymorphism types
studied, no complexity hardness result has yet been presented.
In particular, this is the case for $k$-symmetric and $k$-cyclic
polymorphisms with odd $k \geq 3$, Maltsev polymorphisms, and Siggers polymorphisms.  Can such hardness results be given?

\item Glancing at the table again, 
in all of the cases where a set polymorphism was considered, there is a wide gap between the upper bound of EXPTIME and 
the lower bound of NP-hard or NL-hard.  Can the gap be narrowed?
A possible next question could be to determine if 
deciding the presence of a general set polymorphism is coNP-hard
or not.

\item A number of 
solution procedures that extend arc consistency
have been proposed in the 
literature~\cite{ChenDalmauGrussien13-ACandfriends}.
For one of them,
\emph{look-ahead arc consistency}~\cite{ChenDalmau04-slaac,ChenDalmauGrussien13-ACandfriends}, 
the metaquestion of deciding whether or not the procedure solves
$\csp(\bbH)$ is known to be in polynomial time
(see footnote 2 in~\cite{ChenDalmau04-slaac}).
One can inquire about the complexity of the corresponding metaquestion for
other such extensions, such as \emph{peek arc consistency}~\cite{BodirskyChen10-peek,ChenDalmauGrussien13-ACandfriends}.

\item A \emph{coset-generating operation} on a set $H$
is an operation of the form $f(x,y,z) = xy^{-1}z$,
where the multiplication and inverse are relative to a 
group structure on $H$.
Can anything be said about the complexity of
deciding the presence of coset-generating polymorphism?
Such polymorphisms 
are known to imply tractability of the CSP; 
indeed, each coset-generating operation is a Maltsev operation.

Relatedly, a structure $\bbH$ having a
group polymorphism (that is, a polymorphism that is a binary operation
giving a group structure on $H$) can be readily verified
to have a coset-generating polymorphism.
What is the complexity of deciding if a structure $\bbH$ 
has a group polymorphism?

One can also ask these questions for restricted classes of groups.

\item 
Let us here say that a polymorphism is non-trivial
if it is not a projection.
What is the complexity of deciding if a given
structure has a non-trivial polymorphism?
The following observations can be made.

\begin{lemma} 
It is decidable to determine 
if a structure has a non-trivial idempotent polymorphism.
\end{lemma}

\begin{proof}  Let $\bbH$ be a finite structure. We show that if it admits a non-trivial idempotent polymorphism, it has one of arity at most $\max(3,|H|)$. Let $f$ be such a polymorphism with minimal arity. Then if we identify any two variables, we obtain a projection. By Swierczkowski's lemma~\cite{Swie60-independently}, if $f$ has arity 4 or more, then in fact, there exists a unique $i$ such that $f$ projects onto the $i$-th coordinate when we apply it to a tuple containing a repetition. In particular, if the arity is greater than $|H|$, $f$ is a projection, a contradiction. \end{proof}

Observe that this lemma implies the decidability
of determining presence of a non-trivial polymorphism:
one can first check for a non-trivial polymorphism of arity $1$;
if there is none, then every polymorphism is idempotent, 
and one can then invoke the algorithm of this lemma.

Relatedly, one can inquire about the complexity of 
deciding if a given structure has a polymorphism
that is not essentially unary.

\item In this article, we focused on the \emph{explicit representation}
of relational structures, where each relation is specified
by an explicit listing of its tuples.
In some contexts, however, it is natural to assume that the
second structure $\bbH$ of each CSP instance $(\bbG, \bbH)$
has its relations specified according to other representations
(see the discussion in~\cite{ChenGrohe10-succinct}).
Alternative representations have been considered in the literature;
for example, Marx~\cite{Marx11-csptruthtables} studied \emph{truth table representation},
and Chen and Grohe~\cite{ChenGrohe10-succinct} studied two representations
which they called 
\emph{generalized DNF representation} and
\emph{decision diagram representation}.
The latter two representations are more succinct than the
explicit representation, and many of the questions that we
have studied and discussed can be investigated for these representations.
For instance,
for each of these representations,
one can consider the complexity of any metaquestion 
where the input is a structure $\bbH$ under the representation.
One can also investigate whether or not uniform 
polynomial-time algorithms exist for various classes
of structures, under these succinct representations;
this was considered briefly by Chen and Grohe~\cite[Section 5]{ChenGrohe10-succinct},
but it seems that the theory could be developed much further.

\end{itemize}

\paragraph{\bf Acknowledgements.} The authors wish to express their gratitude to Matt Valeriote for numerous useful discussions and comments.
The first author was supported by the
Spanish Project MINECO COMMAS TIN2013-46181-C2-R, Basque Project GIU15/30, and Basque Grant UFI11/45. The second author acknowledges the support of FRQNT and NSERC. The authors are also grateful to the anonymous referees for their judicious comments.

\bibliographystyle{plain}
\bibliography{hubiebib-new}

\begin{thebibliography}{10}

\bibitem{ABISV09-refining}
E.~Allender, M.~Bauland, N.~Immerman, H.~Schnoor, and H.~Vollmer.
\newblock {The Complexity of Satisfiability Problems: Refining Schaefer's
  Theorem.}
\newblock {\em Journal of Computer and System Sciences}, 75(4):245--254, 2009.

\bibitem{BartoKozik09-boundedwidth}
L.~Barto and M.~Kozik.
\newblock Constraint satisfaction problems of bounded width.
\newblock In {\em Proceedings of FOCS'09}, 2009.

\bibitem{Barto14-collapse}
Libor Barto.
\newblock The collapse of the bounded width hierarchy.
\newblock {\em Journal of Logic and Computation}, 2014.

\bibitem{BartoKozik14-localconsistency}
Libor Barto and Marcin Kozik.
\newblock Constraint satisfaction problems solvable by local consistency
  methods.
\newblock {\em J. {ACM}}, 61(1):3, 2014.

\bibitem{BartoKozikStanovsky15-maltsev-lack-solvability}
Libor Barto, Marcin Kozik, and David Stanovsk{\'y}.
\newblock Mal'tsev conditions, lack of absorption, and solvability.
\newblock {\em Algebra universalis}, 74(1):185--206, 2015.

\bibitem{BartoOprsalPinsker15-wonderland}
Libor Barto, Jakub Oprsal, and Michael Pinsker.
\newblock The wonderland of reflections.
\newblock {\em CoRR}, abs/1510.04521, 2015.

\bibitem{BIMMVW10-fewsubpowers}
J.~Berman, P.~Idziak, P.~Markovic, R.~McKenzie, M.~Valeriote, and R.~Willard.
\newblock Varieties with few subalgebras of powers.
\newblock {\em Transactions of the American Mathematical Society},
  362(3):1445--1473, 2010.

\bibitem{BodirskyChen07-oligomorphic}
Manuel Bodirsky and Hubie Chen.
\newblock Oligomorphic clones.
\newblock {\em Algebra Universalis}, 57(1):109--125, 2007.

\bibitem{BodirskyChen09-qualitative}
Manuel Bodirsky and Hubie Chen.
\newblock Qualitative temporal and spatial reasoning revisited.
\newblock {\em Journal of Logic and Computation}, 19(6):1359--1383, 2009.

\bibitem{BodirskyChen10-peek}
Manuel Bodirsky and Hubie Chen.
\newblock Peek arc consistency.
\newblock {\em Theoretical Computer Science}, 411(2):445--453, 2010.

\bibitem{BBCJK09-qcsp}
Ferdinand B\"{o}rner, Andrei~A. Bulatov, Hubie Chen, Peter Jeavons, and
  Andrei~A. Krokhin.
\newblock The complexity of constraint satisfaction games and {QCSP}.
\newblock {\em Information and Computation}, 207(9):923--944, 2009.

\bibitem{BovaChenValeriote13-generic}
Simone Bova, Hubie Chen, and Matthew Valeriote.
\newblock Generic expression hardness results for primitive positive formula
  comparison.
\newblock {\em Inf. Comput.}, 222:108--120, 2013.

\bibitem{BulatovJeavonsKrokhin05-finitealgebras}
A.~Bulatov, P.~Jeavons, and A.~Krokhin.
\newblock {Classifying the Complexity of Constraints using Finite Algebras}.
\newblock {\em SIAM Journal on Computing}, 34(3):720--742, 2005.

\bibitem{BulatovDalmau06-maltsev}
Andrei Bulatov and Victor Dalmau.
\newblock {A Simple Algorithm for Mal'tsev Constraints}.
\newblock {\em SIAM Journal of Computing}, 36(1):16--27, 2006.

\bibitem{BulatovJeavons00-commutative-conservative}
Andrei Bulatov and Peter Jeavons.
\newblock Tractable constraints closed under a binary operation.
\newblock Technical Report PRG-TR-12-00, Oxford University Computing
  Laboratory, 2000.

\bibitem{Bulatov11-conservative-csp}
Andrei~A. Bulatov.
\newblock Complexity of conservative constraint satisfaction problems.
\newblock {\em {ACM} Trans. Comput. Log.}, 12(4):24, 2011.

\bibitem{DBLP:conf/dagstuhl/BulatovKL08}
Andrei~A. Bulatov, Andrei~A. Krokhin, and Benoit Larose.
\newblock Dualities for constraint satisfaction problems.
\newblock In {\em Complexity of Constraints - An Overview of Current Research
  Themes [Result of a Dagstuhl Seminar].}, pages 93--124, 2008.

\bibitem{BulatovValeriote08-recent-results-csp}
Andrei~A. Bulatov and Matthew Valeriote.
\newblock Recent results on the algebraic approach to the {CSP}.
\newblock In {\em Complexity of Constraints - An Overview of Current Research
  Themes [Result of a Dagstuhl Seminar].}, pages 68--92, 2008.

\bibitem{Frenchguypaper}
C.~Carbonnel.
\newblock The meta-problem for conservative mal'tsev constraints.
\newblock In {\em Proceedings of the Thirtieth AAAI Conference on Artificial
  Intelligence (AAAI-16)}, 2016.

\bibitem{KroCar-symmetric}
Catarina Carvalho and Andrei Krokhin.
\newblock On algebras with many symmetric operations, 2016.
\newblock arXiv:1406.5061.

\bibitem{Chen05-expressive}
Hubie Chen.
\newblock The expressive rate of constraints.
\newblock {\em Annals of Mathematics and Artificial Intelligence},
  44(4):341--352, 2005.

\bibitem{Chen11-qcspandpgp}
Hubie Chen.
\newblock Quantified constraint satisfaction and the polynomially generated
  powers property.
\newblock {\em Algebra Universalis}, 65:213--241, 2011.

\bibitem{Chen12-meditations}
Hubie Chen.
\newblock Meditations on quantified constraint satisfaction.
\newblock In Robert Constable and Alexandra Silva, editors, {\em Logic and
  Program Semantics}, volume 7230 of {\em Lecture Notes in Computer Science},
  pages 35--49. Springer Berlin / Heidelberg, 2012.

\bibitem{ChenDalmau04-slaac}
Hubie Chen and Victor Dalmau.
\newblock {({S}mart) Look-Ahead Arc Consistency and the Pursuit of {CSP}
  Tractability}.
\newblock In {\em Principles and Practice of Constraint Programming - CP 2004},
  Lecture Notes in Computer Science. Springer-Verlag, 2004.

\bibitem{ChenDalmauGrussien13-ACandfriends}
Hubie Chen, V\'{\i}ctor Dalmau, and Berit Gru{\ss}ien.
\newblock Arc consistency and friends.
\newblock {\em J. Log. Comput.}, 23(1):87--108, 2013.

\bibitem{ChenGrohe10-succinct}
Hubie Chen and Martin Grohe.
\newblock Constraint satisfaction with succinctly specified relations.
\newblock {\em Journal of Computer and System Sciences}, 76(8):847--860, 2010.

\bibitem{MR3146708}
V{\'{\i}}ctor Dalmau and Andrei Krokhin.
\newblock Robust satisfiability for {CSP}s: hardness and algorithmic results.
\newblock {\em ACM Trans. Comput. Theory}, 5(4):Art. 15, 25, 2013.

\bibitem{DalmauPearson99-width1}
Victor Dalmau and Justin Pearson.
\newblock {Closure Functions and Width 1 Problems}.
\newblock In {\em CP 1999}, pages 159--173, 1999.

\bibitem{Feder01classification}
T.~Feder.
\newblock Classification of homomorphisms to oriented cycles and k-partite
  satisfiability.
\newblock {\em SIAM J. Discrete Math}, 14:471--480, 2001.

\bibitem{FederVardi99-structure}
T.~Feder and M.~Vardi.
\newblock The computational structure of monotone monadic {SNP} and constraint
  satisfaction: {A} study through {D}atalog and group theory.
\newblock {\em {SIAM} Journal on Computing}, 28:57--104, 1999.

\bibitem{DBLP:journals/ai/GreenC08}
Martin~J. Green and David~A. Cohen.
\newblock Domain permutation reduction for constraint satisfaction problems.
\newblock {\em Artif. Intell.}, 172(8-9):1094--1118, 2008.

\bibitem{IMMVW10-tractabilityfewsubpowers}
P.~Idziak, P.~Markovic, R.~McKenzie, M.~Valeriote, and R.~Willard.
\newblock Tractability and learnability arising from algebras with few
  subpowers.
\newblock {\em {SIAM} J. Comput.}, 39(7):3023--3037, 2010.

\bibitem{JeavonsCohenCooper98-consistency}
P.~Jeavons, D.~Cohen, and M.~Cooper.
\newblock Constraints, consistency, and closure.
\newblock {\em Artificial Intelligence}, 101(1-2):251--265, 1998.

\bibitem{Kazda11-csp-binary-conservative}
Alexandr Kazda.
\newblock Csp for binary conservative relational structures.
\newblock {\em arXiv}, abs/1112.1099, 2011.

\bibitem{Kazda11-maltsev-digraphs-have-majority}
Alexandr Kazda.
\newblock Maltsev digraphs have a majority polymorphism.
\newblock {\em European Journal of Combinatorics}, 32(3):390 -- 397, 2011.

\bibitem{Kearnes2014}
Keith Kearnes, Petar Markovi{\'{c}}, and Ralph McKenzie.
\newblock Optimal strong mal'cev conditions for omitting type 1 in locally
  finite varieties.
\newblock {\em Algebra universalis}, 72(1):91--100, 2014.

\bibitem{KozikKrokhinValerioteWillard15-characterizations}
Marcin Kozik, Andrei Krokhin, Matt Valeriote, and Ross Willard.
\newblock Characterizations of several maltsev conditions.
\newblock 2015.
\newblock Preprint.

\bibitem{DBLP:conf/innovations/KunOTYZ12}
G{\'{a}}bor Kun, Ryan O'Donnell, Suguru Tamaki, Yuichi Yoshida, and Yuan Zhou.
\newblock Linear programming, width-1 csps, and robust satisfaction.
\newblock In {\em Innovations in Theoretical Computer Science 2012, Cambridge,
  MA, USA, January 8-10, 2012}, pages 484--495, 2012.

\bibitem{LaroseTesson09-hardness}
Benoit Larose and Pascal Tesson.
\newblock Universal algebra and hardness results for constraint satisfaction
  problems.
\newblock {\em Theoretical Computer Science}, 410(18):1629--1647, 2009.

\bibitem{Marx11-csptruthtables}
D{\'a}niel Marx.
\newblock Tractable structures for constraint satisfaction with truth tables.
\newblock {\em Theory of Computing Systems}, 48:444--464, 2011.

\bibitem{Opatrny79-betweenness}
Jaroslav Opatrny.
\newblock Total ordering problem.
\newblock {\em {SIAM} J. Comput.}, 8(1):111--114, 1979.

\bibitem{Schaefer78-satisfiability}
T.~J. Schaefer.
\newblock The complexity of satisfiability problems.
\newblock In {\em Proceedings of STOC'78}, pages 216--226, 1978.

\bibitem{Siggers2010}
Mark~H. Siggers.
\newblock A strong mal'cev condition for locally finite varieties omitting the
  unary type.
\newblock {\em Algebra universalis}, 64(1):15--20, 2010.

\bibitem{Swie60-independently}
S~\'{S}wierczkowski.
\newblock Algebras which are independently generated by every n elements.
\newblock {\em Fund. Math. 49}, pages 93–--104, 1960/1961.

\bibitem{Szendrei92-survey}
A.~Szendrei.
\newblock A survey on strictly simple algebras and minimal varieties.
\newblock In {\em Research and Exposition in Mathematics, Heldermann Verlag},
  pages 209--239, 1992.

\bibitem{Valeriote09-subalgebra-intersection}
M.~Valeriote.
\newblock A subalgebra intersection property for congruence distributive
  varieties.
\newblock {\em the Canadian Journal of Mathematics}, 61(2):451--464, 2009.

\end{thebibliography}


\end{document}